\documentclass[oneside, 11pt, reqno]{amsart}
\usepackage[left=1in, right=1in, top=0.7in, bottom=0.7in]{geometry}
\usepackage[usenames,dvipsnames]{xcolor}
\usepackage{graphicx,enumitem}

\usepackage{multirow}
\newcommand{\figcaption}[1]{
	\addtocounter{figure}{1}
	{\center Figure \arabic{figure}. #1\\ }
	\addcontentsline{lof}{figure}{#1}
}

\DeclareMathOperator\sgn{sgn}

\newcommand{\ox}{\overline{x}}    \newcommand{\ux}{\underline{x}}
 
\newcommand{\old}{\overline{\ld}} \newcommand{\uld}{\underline{\ld}}

\newcommand{\oy}{\overline{y}}    \newcommand{\uy}{\underline{y}}
\newcommand{\oS}{\overline{S}}    \newcommand{\uS}{\underline{S}}

\newcommand{\ld}{\lambda}

\newcommand{\define}[1]{{\textbf{#1}}}

\newcommand{\Zitkovic}[1]{{\v Z}itkovi\'c}
\newcommand{\Sirbu}[1]{S\^\i rbu}

\newcommand\hc{{\hat{c}}}

\newcommand\hm{{\hat{m}}}

\newcommand\hs{{\hat{s}}}

\newcommand\hx{{\hat{x}}}

\newcommand\sC{{\mathcal C}}

\newcommand\sE{{\mathcal E}}

\newcommand\bE{{\mathbb E}}

\newcommand\hH{{\hat{H}}}

\newcommand\bN{{\mathbb N}}

\newcommand\bP{{\mathbb P}}

\newcommand\bR{{\mathbb R}}

\newcommand\sS{{\mathcal S}}

\newcommand\tS{{\tilde{S}}}

\newcommand\hW{{\hat{W}}}

\numberwithin{equation}{section}
\theoremstyle{plain}                % title and number in bold, text italic
\newtheorem{theorem}{Theorem}[section]
\newtheorem{lemma}[theorem]{Lemma}
\newtheorem{proposition}[theorem]{Proposition}
\newtheorem{corollary}[theorem]{Corollary}

\theoremstyle{definition}           % title and number in bold, text normal
\newtheorem{definition}[theorem]{Definition}

\newtheorem{assumption}[theorem]{Assumption}
\theoremstyle{remark}
\newtheorem{remark}[theorem]{Remark}

\newcommand{\defref}[1]{Definition~\ref{#1}}
\newcommand{\thmref}[1]{Theorem~\ref{#1}}
\newcommand{\proref}[1]{Proposition~\ref{#1}}

\newcommand{\lemref}[1]{Lemma~\ref{#1}}
\newcommand{\corref}[1]{Corollary~\ref{#1}}
\newcommand{\remref}[1]{Remark~\ref{#1}}
\newcommand{\assref}[1]{Assumption~\ref{#1}}
\usepackage{graphicx}

\begin{document}

\author{Jin Hyuk Choi}
\thanks{Department of mathematics, Ulsan National Institute of Science and Technology; jchoi@unist.ac.kr}
\title{Optimal consumption and investment
  with liquid and illiquid assets}

\maketitle

\begin{abstract}
We consider an optimal consumption/investment problem to maximize expected utility from consumption. In this market model, the investor is allowed
to choose a portfolio which consists of one bond, one liquid risky asset (no transaction costs) and one illiquid risky asset (proportional transaction
costs). We fully characterize the optimal consumption and trading strategies in terms of the solution of the free boundary ODE with an integral constraint. We find an explicit characterization of model parameters for the well-posedness of the problem, and show that the problem is well-posed if and only if there exists a shadow price process. 
Finally, we describe how the investor's optimal strategy is affected by the additional opportunity of trading the liquid risky asset, compared to the simpler model with one bond and one illiquid risky asset. 
\end{abstract}

\bigskip


\bigskip

\section{Introduction}

In the seminal papers \cite{Mer69, Mer71}, Merton formulated and solved the optimal consumption and investment problem in the continuous-time stochastic control framework. Under the assumption that the risky asset price process is a geometric Brownian motion and the investor has a CRRA (constant relative risk aversion) utility function, Merton proved that it is optimal to invest a constant proportion of wealth in the risky asset. Since then, the dynamic optimal consumption/investment problems have been studied by many researchers, and the results extend to very general situations (e.g., \cite{KLS87, KLSX91, KraSch99, KarZit03, Kra04}), under the simplifying assumption of no transaction costs (perfect liquidation).

One type of generalization of these problems is to consider transaction costs which are levied on each transaction.
 Constantinides and Magill \cite{MagCon76} assumed proportional transaction costs in the model of \cite{Mer71}. They intuited that the optimal strategy is to keep the proportion of wealth invested in the risky asset in an interval, by trading the risky asset in a minimal way. Davis and Norman \cite{DavNor90} proved this intuition by formulating the HJB (Hamilton-Jacobi-Bellman) equation. Shreve and Soner \cite{ShrSon94} subsequently complemented the analysis of \cite{DavNor90}, by removing various technical conditions and using the technique of viscosity solutions to clarify the key arguments. Because the solution of the HJB equation is not explicit, except in the case of no transaction costs, the asymptotic analysis for small transaction costs has also been studied (for a single risky asset case, see, e.g., \cite{ShrSon94, JanShr04, Bic12, GerMuhSch11, Cho11}).

The market model in Davis and Norman \cite{DavNor90} and Shreve and Soner \cite{ShrSon94} consists of a single risky asset. Even though the natural extension is to consider a model with multiple risky assets, it is known that transaction costs models with multiple assets are notably harder to analyze than a model with a single risky asset. Consequently, most of the existing results are limited to models with a single risky asset. 

For multiple-asset models, Akian \emph{et al}. \cite{Aki96} prove that the value function is the viscosity solution of a variational inequality. Liu \cite{Liu04} considers the model with exponential utility and independent Brownian motions: In this special case, the multiple-asset problem can be decomposed into a set of single risky asset problems. Muthuraman and Kumar \cite{Mut06} develop a numerical method to solve the multiple-asset problem. Chen and Dai \cite{CheDai13} characterize the shape of the no-trading region in the model with two risky assets. Bichuch and Shreve \cite{BicShr11} prove an asymptotic expansion for small transaction costs, in the market with two futures. Possamai \emph{et al}. \cite{Pos15} prove an asymptotic expansion for small transaction costs for general Markovian risky asset processes. Because a rigorous characterization of the optimal strategies is unknown in the models with multiple assets (\cite{Liu04} is an exception), these papers \cite{Aki96, Mut06, CheDai13, BicShr11, Pos15} focus on asymptotic analysis or some characteristics of the no-trading region.

In this paper, we consider an optimal consumption/investment problem in the market which consists of one bond, one liquid risky asset and one illiquid risky asset. The investor pays proportional transaction costs for trading the illiquid asset, but the other risky asset is perfectly liquid. To study this problem, we employ the \emph{shadow price approach} used in \cite{KalMuh10, GerMuhSch11, ChoSirZit11, GerGua14, Cho11,HerPro15, Guasoni}. The shadow price approach amounts to construct the most unfavorable frictionless market, where the asset price processes lie between the bid and ask prices of the original market. After proving that the constructed frictionless market produces the same expected utility as the original market, we obtain the expressions for the optimal strategies and value function by solving the optimization problem in the frictionless market.% We find a candidate of the shadow price process using the solution of a constraint free boundary ODE, and do the verification.
\footnote{Our analysis does not rely on the dynamic programming principle or the technique of viscosity solutions as in \cite{ShrSon94, BicShr11}.}

For CRRA utility functions and infinite time horizon, we fully characterize the value function and the optimal consumption/trading strategies in terms of the solution of a free boundary ODE. We also provide an explicit equivalent condition on the model parameters for  the finiteness of the value function (well-posedness of the problem), and prove that there exists a shadow price process if and only if the problem is well-posed. The approach of this paper is close to that of \cite{ChoSirZit11}, but the structure of the ODE turns out to be more complicated in the current model with the liquid risky asset.\footnote{
To be specific, the form of the ODE in \cite{ChoSirZit11} is 
\begin{equation}\label{one}
g'(x)=\tfrac{P(x,g(x))}{Q(x,g(x))},
\end{equation}
but the ODE in this paper has the form of 
\begin{equation}\label{two}
g'(x)=\tfrac{2C(x,g(x))}{-B(x,g(x))+\sqrt{B(x,g(x))^2-4 A(x,g(x)) C(x,g(x))}}\textrm{  or  }\tfrac{2C(x,g(x))}{-B(x,g(x))-\sqrt{B(x,g(x))^2-4 A(x,g(x)) C(x,g(x))}},
\end{equation}
 where $Q(x,y), P(x,y), A(x,y), B(x,y), C(x,y)$ are quadratic in $x$ and $y$. This paper is technically more involved, because the form of the ODE in \eqref{two} is more complicated than \eqref{one}, and we also need to determine which ODE to choose in \eqref{two} for various parameter conditions.
See Section 6 for details.

}

Due to the complicated nature of the problem, the asymptotic analysis for small transaction costs is useful to understand the optimal behavior. Comparing our model with the model without the liquid risky asset, we describe how the investor's optimal consumption and illiquid asset trading strategies are affected by the additional opportunity of trading the liquid risky asset. Also, we describe how the optimal trading strategy for the liquid risky asset is affected by the presence of the illiquid asset, compared to the frictionless Merton problem.

%Our analysis do not rely on technical assumptions for the market parameters, e.g.,  the size of transaction costs (\cite{GerGua14, GerMuhSch11, Pos15, HerPro15, Guasoni}), or the location of Merton line (\cite{KalMuh10, GerMuhSch11}). 

Our model is similar to the models in \cite{Dai11, Guasoni, Hobson16}. Dai \emph{et al} \cite{Dai11} consider a model with a finite horizon and position constraints, and they characterize the trading boundaries. 
%Hobson and Zhu \cite{Hop14} investigate an extreme case of the model in this paper: They consider a model with infinite transaction costs for buying and no transaction costs for selling. 
Guasoni and Bichuch \cite{Guasoni} consider the problem of maximizing the long-term growth rate. Under the assumption of small transaction costs, they solve the problem using the \emph{shadow price approach}, and prove an asymptotic expansion result. In parallel with our work, Hobson \emph{et al}. \cite{Hobson16} recently consider a similar problem as in this paper and solve it by studying the HJB equation of the primal optimization problem. They also provide an explicit characterization of the well-posedness of the problem.

The remainder of the paper is organized as follows: Section 2 describes the model. In Section 3, we explain the shadow price approach, and heuristically derive a free boundary ODE from the property of the shadow price process. In Section 4, we state and prove the main results: The expression of the optimal strategy and value function, the existence of a shadow price process, and the criteria for the well-posedness of the problem are given. Section 5 describes some properties of the optimal strategy and provides intuitive explanations for them. 
Finally, Section 6 is mainly devoted to prove the existence of a smooth solution to the free boundary ODE with an integral constraint.

\section{The Model}

The market model we consider consists of one zero-interest bond\footnote{
%If the interest rate of the bond is a nonzero constant $r>0$, the budget constraint \eqref{budget} should be
%$$d(e^{rt} \varphi_t^{(0)} + \varphi_t^{(2)} S_t^{(2)}) = re^{rt}\varphi_t^{(0)}dt+  \varphi_t^{(2)}dS_t^{(2)} - \overline{S}^{(1)}_t d(\varphi_t^{(1)})^\uparrow +\underline{S}^{(1)}_t d(\varphi_t^{(1)})^\downarrow - c_t dt,$$
%which can be rewritten in terms of $S^{(1)*}_t=e^{-rt}S^{(1)}_t,\,S^{(2)*}_t=e^{-rt}S^{(2)}_t\,c^*_t=e^{-rt}c_t$ as
%$$d(\varphi_t^{(0)} + \varphi_t^{(2)} S_t^{(2)*}) =  \varphi_t^{(2)}dS_t^{(2)*} - \overline{S}^{(1)*}_t d(\varphi_t^{(1)})^\uparrow +\underline{S}^{(1)*}_t d(\varphi_t^{(1)})^\downarrow - c^*_t dt.$$
%Then one can easily check that the market with parameters $(\delta,\mu_1,\mu_2,\sigma_1,\sigma_2,r)$ and the market with parameters $(\delta-r p ,\mu_1-r,\mu_2-r,\sigma_1,\sigma_2,0)$ produce the same expected utilities. Therefore, we assume $r=0$ without loss of generality.
The case with non-zero constant interest rate can be transformed to the case with zero interest rate.
}
 and two risky assets, whose price processes $S^{(1)}$ and $S^{(2)}$ are given by
\begin{equation}
\begin{split}\label{S}
dS^{(i)} = S^{(i)} (\mu_i dt + \sigma_i dB_t^{(i)}), \quad S_0^{(i)}>0, \quad i=1,2.\\
\end{split}
\end{equation}
Here, $B^{(1)}$ and $B^{(2)}$ are standard Brownian motions with correlation $\rho\in (-1,1)$, and the parameters $\mu_i$ and 
$\sigma_i$ are positive constants.
The information structure is given by the augmented filtration generated by $B^{(1)}$ and $B^{(2)}$. We assume that $S^{(2)}$ can be traded without transaction costs, but proportional transaction costs are imposed whenever an investor trades $S^{(1)}$. We call $S^{(1)}$ an illiquid asset and $S^{(2)}$ a liquid asset. To be specific, there are constants $\overline{\lambda}>0$ and $\underline{\lambda}\in (0,1)$ such that the investor pays $\overline{S}^{(1)}_t:=(1+\overline{\lambda})S_t^{(1)}$ for one share of the illiquid asset, but only gets $\underline{S}^{(1)}_t:=(1-\underline{\lambda})S_t^{(1)}$ for one share of the illiquid asset.

Let the investor initially hold $\eta_0$ shares of the bond, $\eta_1$ shares of illiquid asset, and $\eta_2$ shares of liquid asset. In terms of notation, let the triple $(\varphi_t^{(0)},\varphi_t^{(1)}, \varphi_t^{(2)})$ represent the number of shares in the bond and the two risky assets at time $t$, and let $c_t$ be the consumption rate. In order to incorporate the possibility of the initial jump, we distinguish $(\varphi_{0-}^{(0)},\varphi_{0-}^{(1)}, \varphi_{0-}^{(2)})$ and $(\varphi_{0}^{(0)},\varphi_{0}^{(1)}, \varphi_{0}^{(2)})$. The processes are right-continuous after that. We set $(\varphi_{0-}^{(0)},\varphi_{0-}^{(1)}, \varphi_{0-}^{(2)})=(\eta_0,\eta_1,\eta_2)$.

\begin{definition}\label{financeable}
$\sC$ is a set of nonnegative, right-continuous, and locally integrable optional processes, such that $c \in \sC$ if there exist right-continuous optional processes $(\varphi^{(0)},\varphi^{(1)}, \varphi^{(2)})$ which satisfy the following three conditions:\\ 
(i) $\varphi^{(1)}$ is of finite variation a.s.\\
(ii) (Admissibility) The liquidation value is always nonnegative, i.e., 
\begin{equation}
\begin{split}\label{adm}
\varphi_t^{(0)}+\underline{S}_t^{(1)} (\varphi_t^{(1)})^+ -\overline{S}_t^{(1)} (\varphi_t^{(1)})^- + S_t^{(2)}\varphi_t^{(2)} \geq 0, \quad t\geq 0.
\end{split}
\end{equation}
(iii) (Budget constraint) The consumption stream is financeable, i.e.,
\begin{equation}
\begin{split}\label{budget}
\varphi_t^{(0)} + \varphi_t^{(2)} S_t^{(2)} =  \eta_0 +  \eta_2 S_0^{(2)}+ \int_0^t \varphi_u^{(2)}dS_u^{(2)} - \int_0^t \overline{S}^{(1)}_u d(\varphi_u^{(1)})^\uparrow +\int_0^t \underline{S}^{(1)}_u d(\varphi_u^{(1)})^\downarrow - \int_0^t c_u du, 
\end{split}
\end{equation}
where $(\varphi_t^{(1)})^\uparrow$ and $(\varphi_t^{(1)})^\downarrow$ are the cumulative numbers of illiquid asset bought and sold up to time $t$.
\end{definition}

For the initial admissibility, we assume that 
$$
\eta_0 + \uS_t^{(1)} (\eta_1)^+ - \oS_t^{(1)} (\eta_1)^-  + S_t^{(2)} \eta_2 \geq 0.
$$
%\begin{remark}
%Since we consider utility maximization problem by consumption, the inequality in \defref{financeable} (iii) is tight for the optimal strategy.
%\end{remark}

For $p\in (-\infty,1)\setminus \{ 0\}$, we consider the \define{utility function}
$U:[0,\infty)\to [-\infty,\infty)$ of the power (CRRA) type. It is defined
for $c\geq 0$ by
\[ U(c) = \frac{c^p}{p}, 
\text{ and }\ U(0)= \begin{cases} 0, & p>0,\\ -\infty, & p\leq 0
\end{cases}
\]

%\begin{remark}
%For the logarithmic utility case ($p=0$), one can check that the optimization problem here is equivalent to the optimization problem appears in (reference) after the re-parametrization.
%\end{remark}

Our goal is to analyze the optimal consumption and investment problem
\begin{equation}
\begin{split}\label{primal problem}
\sup_{c\in \sC} \bE\Big[\int_0^\infty e^{-\delta t} U(c_t)dt \Big],
\end{split}
\end{equation}
where the constant $\delta$ is the impatience rate.

For convenience, let $q:=p/(1-p)$.

\begin{remark}\label{finite value}
The optimization problem \eqref{primal problem} is ill-posed (i.e., produces infinite value) for some parameter conditions:\\ 
(1) If $\delta \leq  \tfrac{q \mu_2}{2 \sigma_2^2}$, then the value of the optimization problem without trading the illiquid asset is infinity (see Theorem 2.1 in \cite{DavNor90}).\\
(2) If $\delta \leq \tfrac{q(2\mu_1(1+q)-\sigma_1^2)}{2(1+q)^2}$, then the value of the optimization problem without trading the liquid asset is infinity (see Proposition 6.1 in \cite{ChoSirZit11}).
\end{remark}

\begin{remark}\label{special cases}
The special cases $\mu_1=\frac{\rho \mu_2 \sigma_1}{\sigma_2}$ or $\mu_2= \frac{\rho \sigma_1 \sigma_2}{1+q}$ are covered by the result of \cite{ChoSirZit11} regarding the single risky asset case:\\
(1) Suppose that $\mu_1=\frac{\rho \mu_2 \sigma_1}{\sigma_2}$. If there is no transaction costs, then it is optimal to hold $0$ shares of $S^{(1)}$. This implies that the optimal strategy in the original transaction costs model never trades the asset $S^{(1)}$, and the problem is reduced to the frictionless model with $S^{(2)}$ only.\\
(2) Suppose that $\mu_2= \frac{\rho \sigma_1 \sigma_2}{1+q}$. One can check that the ODE in our model reduces to the ODE in the single risky asset model in \cite{ChoSirZit11}. See Appendix for details and financial interpretation.
\end{remark}

Based on \remref{finite value} and \remref{special cases}, we impose the following standing assumption throughout the rest of the paper.

\begin{assumption}\label{ass}
The parameters of the optimization problem satisfy the following conditions:
(1) $\delta >\tfrac{q \mu_2}{2 \sigma_2^2}$ and $\delta >\tfrac{q(2\mu_1(1+q)-\sigma_1^2)}{2(1+q)^2}$. \\%Otherwise, the optimization problem is ill-posed. 
(2) $\mu_1 \neq \tfrac{\rho \mu_2 \sigma_1}{\sigma_2} $ and $\mu_2 \neq \tfrac{\rho \sigma_1 \sigma_2}{1+q}$. %Otherwise, the optimization problem reduces to the problem in \cite{ChoSirZit11}.
\end{assumption}

\section{Heuristics with shadow price process}
In this section, we explain the so called  shadow price approach in this context, and heuristically derive a free boundary ODE from properties of the shadow price process.

\subsection{Shadow price approach}
In the shadow price approach (see \cite{KalMuh10, GerMuhSch11, ChoSirZit11, GerGua14, Cho11, Guasoni}), the original transaction cost problem is solved by constructing a suitable frictionless (i.e., no transaction costs) market model. We first define the set of consistent price processes, and a set of financeable consumptions in the frictionless market, in \defref{consistent}. Then the definition of the shadow price process is given in \defref{shadow_define}.

\begin{definition}\label{consistent}
(1) The set of consistent price processes $\sS$ is defined as
\begin{equation}
\begin{split}
\sS=\Big\{ \tS: \, \tS \textrm{  is an Ito-process, and  }  \uS^{(1)}_t \leq \tS_t \leq \oS^{(1)}_t \textrm{  for all }t\geq 0, \textrm{ a.s.} \Big\}
\end{split}
\end{equation}
(2) For each $\tS \in \sS$, $\sC(\tS)$ is a set of financeable consumptions in the frictionless market with risky assets $\tS$ and $S^{(2)}$. To be specific, the set $\sC(\tS)$ is defined as a set of nonnegative, locally integrable progressively measurable processes $c$, such that $c \in \sC(\tS)$ if there exist progressively measurable processes $(\varphi^{(0)},\varphi^{(1)},\varphi^{(2)})$ which satisfy the following two conditions:  \\
(i) (Admissibility) Total wealth ($W$ for notation) is always nonnegative, i.e., 
\begin{equation}
\begin{split}\label{shadow_admissible}
W_t:=\varphi_t^{(0)}+\tS_t \varphi_t^{(1)} + S_t^{(2)}\varphi_t^{(2)} \geq 0, \quad t\geq 0.
\end{split}
\end{equation}
(ii) (Budget constraint) The consumption stream is financeable, i.e.,
\begin{equation}
\begin{split}\label{shadow_financeable}
W_t = W_{0-} +    \int_0^t \varphi_u^{(1)} d\tS_t+ \int_0^t \varphi_u^{(2)}dS_u^{(2)} - \int_0^t c_u du, \quad t\geq 0.
\end{split}
\end{equation}
\end{definition}

The connection between the original transaction cost problem and the collection of frictionless problems is described in the following proposition. It is a simple translation of Proposition 2.2 in \cite{ChoSirZit11}.

\begin{proposition}\label{shadow_connection} The following two statements hold.\\
(1) For each $\tS \in \sS$, 
\begin{equation}
\begin{split}\label{shadow_ineq}
\sup_{c\in \sC} \bE\Big[\int_0^\infty e^{-\delta t} U(c_t)dt \Big] \leq \sup_{c\in \sC(\tS)}  \bE\Big[\int_0^\infty e^{-\delta t} U(c_t)dt \Big].
\end{split}
\end{equation}
(2) Given $\tS \in \sS$, let $\hc \in \sC(\tS)$ solve the frictionless optimization problem, i.e.,
\begin{equation}
\begin{split}\label{shadow pro}
\bE\Big[\int_0^\infty e^{-\delta t} U(\hc_t)dt \Big]=\sup_{c\in \sC(\tS)}  \bE\Big[\int_0^\infty e^{-\delta t} U(c_t)dt \Big],
\end{split}
\end{equation}
with $(\hat{\varphi}^{(0)},\hat{\varphi}^{(1)},\hat{\varphi}^{(2)})$ which satisfies the budget constraint \eqref{shadow_financeable}. Assume that\\
(i) $\hat{\varphi}^{(1)}$ is a right-continuous process of finite variation,\\
(ii) $(\hat{\varphi}^{(0)},\hat{\varphi}^{(1)},\hat{\varphi}^{(2)})$ satisfies \eqref{adm},\\
(iii) $d(\hat{\varphi}^{(1)}_t)^\uparrow=1_{\{ \tS_t=\oS_t \}}d(\hat{\varphi}^{(1)}_t)^\uparrow$  and  $d(\hat{\varphi}^{(1)}_t)^\downarrow= 1_{\{ \tS_t=\uS_t \} } d(\hat{\varphi}^{(1)}_t)^\downarrow$.\\
(iv) $\hc, \hat{\varphi}^{(0)},\hat{\varphi}^{(1)},\hat{\varphi}^{(2)}$ are continuous processes except for a possible initial jump at $t=0-$.\\
Then $\hc\in \sC$, and $\hc$ solves the original optimization problem \eqref{primal problem}, i.e., 
\begin{equation}
\begin{split}\label{shadow_eq}
\bE\Big[\int_0^\infty e^{-\delta t} U(\hc_t)dt \Big]=\sup_{c\in \sC}  \bE\Big[\int_0^\infty e^{-\delta t} U(c_t)dt \Big].
\end{split}
\end{equation}
\end{proposition}
\begin{proof}
(1) For any $c\in \sC$, there exists $(\varphi^{(0)},\varphi^{(1)},\varphi^{(2)})$ which satisfies \eqref{budget}. 
\begin{displaymath}
\begin{split}
\varphi_t^{(0)} + \varphi_t^{(2)} S_t^{(2)} &=  \eta_0 +  \eta_{2}S_0^{(2)}+ \int_0^t \varphi_u^{(2)}dS_u^{(2)} - \int_0^t \overline{S}^{(1)}_u d(\varphi_u^{(1)})^\uparrow +\int_0^t \underline{S}^{(1)}_u d(\varphi_u^{(1)})^\downarrow - \int_0^t c_u du\\
&\leq \eta_0 +  \eta_2 S_0^{(2)}+ \int_0^t \varphi_u^{(2)}dS_u^{(2)} - \int_0^t \tS_u d\varphi_u^{(1)} - \int_0^t c_u du,\\
\end{split}
\end{displaymath}
where the inequality is due to $\tS\in\sS$. Then the integration-by-parts formula produces
$$\varphi_t^{(0)} +\varphi_{t}^{(1)}\tS_t+ \varphi_t^{(2)} S_t^{(2)}\leq \eta_0 + \eta_1 \tS_0+  \eta_2 S_0^{(2)}+\int_0^t  \varphi_u^{(1)}d\tS_u+ \int_0^t \varphi_u^{(2)}dS_u^{(2)} - \int_0^t c_u du.
$$
Therefore, if we define $\tilde{\varphi}^{(0)}$ as
$$\tilde{\varphi}_t^{(0)}:=\eta_0 + \eta_1 \tS_0+  \eta_2 S_0^{(2)}+\int_0^t  \varphi_u^{(1)}d\tS_u+ \int_0^t \varphi_u^{(2)}dS_u^{(2)} - \int_0^t c_u du -\varphi_{t}^{(1)}\tS_t- \varphi_t^{(2)} S_t^{(2)}, $$
then $\tilde{\varphi}^{(0)} \geq \varphi^{(0)}$ and \eqref{shadow_financeable} is satisfied with $(\tilde{\varphi}^{(0)},\varphi^{(1)},\varphi^{(2)}, c)$. We also check \eqref{shadow_admissible},
$$0\leq \varphi_t^{(0)}+\underline{S}_t^{(1)} (\varphi_t^{(1)})^+ -\overline{S}_t^{(1)} (\varphi_t^{(1)})^- + S_t^{(2)}\varphi_t^{(2)} \leq \tilde{\varphi}_t^{(0)} +\varphi_{t}^{(1)}\tS_t+ \varphi_t^{(2)} S_t^{(2)}.  $$
Therefore, $c\in \sC(\tS)$ and the inclusion $\sC \in \sC(\tS)$ finishes the proof of (1).

(2) Let $(\hc,\hat{\varphi}^{(0)},\hat{\varphi}^{(1)},\hat{\varphi}^{(2)})$ satisfies the assumptions in the proposition. Then by \eqref{shadow_financeable} and the integration-by-parts formula,
\begin{displaymath}
\begin{split}
\hat{\varphi}^{(0)}_t + \hat{\varphi}^{(2)}_t S_t^{(2)} &= - \hat{\varphi}_t^{(1)} \tS_t + \eta_0 + \eta_1 \tS_0+  \eta_2 S_0^{(2)} + \int_0^t  \hat{\varphi}_u^{(1)}d\tS_u+ \int_0^t \hat{\varphi}_u^{(2)}dS_u^{(2)} - \int_0^t \hc_u du\\
&=  \eta_0 +  \eta_2 S_0^{(2)} - \int_0^t \tS_u d\hat{\varphi}^{(1)}_u + \int_0^t \hat{\varphi}_u^{(2)}dS_u^{(2)}- \int_0^t \hc_u du\\
&=\eta_0 +  \eta_2 S_0^{(2)} - \int_0^t \oS_u d(\hat{\varphi}^{(1)})^\uparrow_u+\int_0^t \uS_u d(\hat{\varphi}^{(1)})^\downarrow_u + \int_0^t \hat{\varphi}_u^{(2)}dS_u^{(2)}- \int_0^t \hc_u du\\
\end{split}
\end{displaymath}
Hence \eqref{budget} is satisfied, and $\hc \in \sC$. Then \eqref{shadow_ineq} and \eqref{shadow pro} imply \eqref{shadow_eq}. 
\end{proof}

\begin{definition}\label{shadow_define}
If $\tS\in\sS$ satisfies the following equality
\begin{equation}
\begin{split}
\sup_{c\in \sC} \bE\Big[\int_0^\infty e^{-\delta t} U(c_t)dt \Big] = \sup_{c\in \sC(\tS)}  \bE\Big[\int_0^\infty e^{-\delta t} U(c_t)dt \Big]<\infty,
\end{split}
\end{equation}
then $\tS$ is called a shadow price process.
\end{definition}

\proref{shadow_connection} (2) implies that we can solve the original transaction costs problem by solving the frictionless problem with shadow price process, and \proref{shadow_connection} (1) says that the shadow price process can be characterized as the solution of the following minimization problem
\begin{equation}
\begin{split}\label{dual}
\inf_{\tS \in \sS} \Big( \sup_{c\in \sC(\tS)}  \bE\Big[\int_0^\infty e^{-\delta t} U(c_t)dt \Big]\Big).
\end{split}
\end{equation}

\subsection{Heuristic derivation of the free boundary ODE}
For the rest of this section, we will heuristically derive a free boundary ordinary differential equation, from the HJB equation for the optimization problem \eqref{dual}. For $\tS\in \sS$, we express $\tS_t=S_t^{(1)} e^{Y_t}$ for an Ito-process $Y$. Because $1-\underline{\lambda} \leq \tS_t/S_t^{(1)} \leq 1+\overline{\lambda}$, we have a natural bound $Y_t\in[\uy,\oy]$, where $\uy:=\ln(1-\underline{\lambda})$ and $\oy:=\ln(1+\overline{\lambda})$. Assume that the dynamics of $Y$ is given by
\begin{equation}
\begin{split}
dY_t = m_t dt + s_{1t} dB^{(1)}_t + s_{2t} dB^{(2)}_t,
\end{split}
\end{equation}
for some processes $m,s_1,s_2$. Then the state price density process $H$ (see, e.g., Remark 5.8, p.~19 in \cite{KarShr98}), in the market with stock prices $\tS$ and $S^{(2)}$, satisfies the stochastic differential equation 
\begin{equation}
\begin{split}
dH_t = - H_t \Big( \theta_1(m_t,s_{1t},s_{2t}) dB^{(1)}_t + \theta_2(m_t,s_{1t},s_{2t}) dB^{(2)}_t\Big), \quad H_0=1,
\end{split}
\end{equation}
where the functions $\theta_1$ and $\theta_2$ are defined as
\begin{equation}
\begin{split}\label{theta}
 \theta_1(m,s_1,s_2)&:=\tfrac{\rho (\sigma_2 s_2 - \mu_2)}{(1-\rho^2) \sigma_2} - \tfrac{\mu_2 s_2 - (m+\mu_1 + s_1 \sigma_1 + \frac{1}{2}(s_1^2+s_2^2))\sigma_2}{(1-\rho^2)\sigma_2 (s_1 +\sigma_1)},\\
 \theta_2(m,s_1,s_2)&:=\tfrac{\mu_2}{\sigma_2}- \rho \, \theta_1(m,s_1,s_2).
\end{split}
\end{equation}

Because the frictionless market model with stock prices $\tS$ and $S^{(2)}$ is complete, the standard duality theory can be applied (see, e.g., Theorem 9.11, p.~141 in \cite{KarShr98})
\begin{equation}
\begin{split}
&\sup_{c\in \sC(\tS)}  \bE\Big[\int_0^\infty e^{-\delta t} U(c_t)dt \Big] \\
&= \inf_{z>0} \Big( \sup_{c}\Big(
 \bE\Big[\int_0^\infty e^{-\delta t} U(c_t)dt \Big]
 + z \Big( (\eta_0 + \tS_0 \eta_1 + S_0^{(2)} \eta_2) - \bE\Big[\int_0^\infty c_t H_t dt \Big]  \Big) \Big)\Big)\\
&=\tfrac{ (\eta_0 + \tS_0 \eta_1 + S_0^{(2)} \eta_2)^p}{p} \Big(\bE\Big[\int_0^\infty e^{-(1+q)\delta t} H_t^{-q}dt \Big]  \Big)^{1-p},
\end{split}
\end{equation}
where $q=p/(1-p)$. Consequently, we can rewrite \eqref{dual} as
\begin{equation}
\begin{split}
\inf_{\tS \in \sS} \Big( \sup_{c\in \sC(\tS)}  \bE\Big[\int_0^\infty e^{-\delta t} U(c_t)dt \Big]\Big) 
=\inf_{Y_0} \Big\{ \tfrac{ (\eta_0 + S_0^{(1)} e^{Y_0} \eta_1 + S_0^{(2)} \eta_2)^p}{p} |w(Y_0)|^{1-p} \Big\}, 
\end{split}
\end{equation}
with
\begin{equation}
\begin{split}\label{w}
w(y):=\inf_{m,s_1,s_2} \Big\{ \sgn (p)\, \bE\Big[\int_0^\infty e^{-(1+q)\delta t} H_t^{-q}dt \Big| Y_0=y\Big] \Big\}.
\end{split}
\end{equation}
The formal HJB equation for \eqref{w} has the following form
\begin{equation}
\begin{split}\label{w_hjb}
\inf_{m,s_1, s_2} \Big\{-\alpha(m,s_1,s_2) w(y) +(m+\beta(m,s_1,s_2))w'(y) + \gamma(s_1,s_2) w''(y) +\sgn{(p)}\Big \} =0, 
\end{split}
\end{equation}
where (with $\theta_1=\theta_1(m,s_1,s_2)$ and $\theta_2=\theta_2(m,s_1,s_2)$)
\begin{equation}
\begin{split}\label{ab}
 \alpha(m,s_1,s_2)&:=(1+q)\delta - \tfrac{q(1+q)}{2} \big(\theta_1^2+\theta_2^2+2\rho \, \theta_1 \theta_2\big),\\
 \beta(m,s_1,s_2)&:=q\big((s_1 + \rho s_2)\theta_1+ (\rho s_1 + s_2)\theta_2\big),\\
 \gamma(s_1,s_2)&:=\tfrac{1}{2}\big(s_1^2+s_2^2+2\rho s_1 s_2\big).
\end{split}
\end{equation}
To incorporate the requirement $Y_t\in [\uy,\oy]$, we turn off the diffusion ($s_{1t}=s_{2t}=0$) whenever $Y_t$ reaches the boundary $\uy$ or $\oy$, and let the drift be the inward direction. 
%The heuristic derivation of the boundary condition in \cite{ChoSirZit11} is still valid here: 
By observing the form of the minimizer in \eqref{w_hjb}, we infer that the boundary condition would be
\begin{equation}
\begin{split}\label{w_boundary}
w''(\uy)=w''(\oy)=\infty.
\end{split}
\end{equation}
To handle this infinite boundary condition and reduce the order of the differential equation, we change variable. Let $x=-w'(y)$ and define the function $g:[\ux,\ox] \mapsto \bR$ as $g(x)=w(y)$, with $\ux=-w'(\oy)$ and $\ox=-w'(\uy)$. With $x$ and $g$, \eqref{w_hjb} is written as
\begin{equation}
\begin{split}\label{g_hjb}
&\inf_{m,s_1, s_2} \Big\{-\alpha(m,s_1,s_2) g(x) -(m+\beta(m,s_1,s_2))x + \gamma(s_1,s_2) \tfrac{x}{g'(x)} +\sgn{(p)}\Big \} =0, \quad x\in [\ux,\ox].
\end{split}
\end{equation}
\eqref{w_boundary} and the relation $dy/dx = -g'(x)/x$ produce a boundary condition and an integral constraint:
\begin{equation}
\begin{split}\label{g_boundary}
g'(\ux)=g'(\ox)=0,\quad \int_{\ux}^{\ox} \tfrac{g'(x)}{x} dx = \oy-\uy.
\end{split}
\end{equation}
As $\ux$ and $\ox$ are not predetermined, \eqref{g_hjb} together with \eqref{g_boundary} is a free boundary problem with an integral constraint. 

\begin{remark}
The purpose of this section is only to derive the free boundary problem which we analyze rigorously in the next section: The arguments in this section is heuristic and not rigorous. 
\end{remark}

\section{The results}

In this section, we first present the existence result for the solution of the free boundary problem that we derived in the previous section. Then we construct the candidate shadow price process $\tS$ using the solution of the free boundary problem. In \lemref{shadow solve}, we solve the optimization problem for the market with the candidate shadow price process. In \thmref{main verify}, we verify that $\tS$ is indeed the shadow price process by checking the conditions in \proref{shadow_connection} (2), and conclude that the optimal solution in \lemref{shadow solve} also solves the original transaction cost problem \eqref{primal problem}. Finally, we provide an explicit characterization of the well-posedness of the problem in \thmref{well-posed thm}.

\subsection{Construction of the shadow price}

The proofs of results related to the free boundary problem are postponed to Section 6 due to their technical nature.

\begin{proposition} \label{FBO}
Assume that the model parameters satisfy one of the following conditions:\\
(i) $p \leq 0$,\\
(ii) $0<p<1$ and $\delta > \tfrac{q}{2(1-\rho^2)} \big( (\tfrac{\mu_1}{\sigma_1})^2+ (\tfrac{\mu_2}{\sigma_2})^2- 2\rho \tfrac{\mu_1 \mu_2}{\sigma_1 \sigma_2} \big) $,\\
(iii) $0<p<1$, $\delta \leq  \tfrac{q}{2(1-\rho^2)} \big( (\tfrac{\mu_1}{\sigma_1})^2+ (\tfrac{\mu_2}{\sigma_2})^2- 2\rho \tfrac{\mu_1 \mu_2}{\sigma_1 \sigma_2} \big) $ and 
$c^*< \ln \big(\tfrac{1+\overline{\lambda}}{1-\underline{\lambda}}\big)$,\\
where $c^*$ is a constant explicitly defined in \defref{c definition}. 

\smallskip
Then, there exist constants $\underline{x}, \overline{x}$ and a function $g \in C^2 [\underline{x},\overline{x}]$ that satisfy following conditions: \\
(1) If $\mu_1>\frac{\rho \mu_2 \sigma_1}{\sigma_2}$, then $0<\ux<\ox$. If $\mu_1<\frac{\rho \mu_2 \sigma_1}{\sigma_2}$, then $\ux<\ox<0$.         \\
(2)  For $x\in [\underline{x}, \overline{x}]$, $g$ satisfies the differential equation 
\begin{equation}
\begin{split}\label{inf_ode}
\inf_{m,s_1, s_2} \Big\{-\alpha(m,s_1,s_2) g(x) -(m+\beta(m,s_1,s_2))x + \gamma(s_1,s_2) \tfrac{x}{g'(x)} +\sgn{(p)}\Big \} =0,
\end{split}
\end{equation}
where the functions $\alpha, \beta, \gamma$ are given in \eqref{theta} and \eqref{ab}.\\
(3) The following boundary/integral conditions are satisfied
	\begin{equation}\label{equ:integral-cond}
	 \begin{split}
	   g'(\ux)=g'(\ox)=0\text{ and }
	\int_{\ux}^{\ox} \tfrac{g'(x)}{x}\,
	dx=\ln(\tfrac{1+\old}{1-\uld}).
	 \end{split}
	\end{equation}
(4) The functions 
$$q\, g(x), \quad q \,  g(x) (g'(x)+1)-(1+q) x g'(x), \quad q\,\big( g(x)-x\, g'(x)\big), \quad \textrm{  and  } \quad g'(x)+1$$  are strictly positive on $[\ux,\ox]$. Recall that $q=p/(1-p)$.\\
(5) $g'(x)/x > 0$ for $x\in (\ux,\ox)$.
\end{proposition}
\begin{proof}
See Section 6.
\end{proof}

We need the next corollary to construct the shadow price process.
\begin{corollary}\label{cor1}
(1) The minimizer $(\hm(x),\hs_1(x),\hs_2(x))$ of \eqref{inf_ode} is well defined on $[\ux,\ox]$.\\
(2) Let $\hat{\alpha}, \hat{\beta},\hat{\gamma},\hat{\theta}_1, \hat{\theta}_2 : [\ux,\ox]\mapsto \bR$ be the composition of the functions $\alpha, \beta, \gamma, \theta_1, \theta_2$ of \eqref{ab} and \eqref{theta} with the optimizers $\hm,\hs_1,\hs_2$ of \eqref{inf_ode}. For instance, $\hat{\alpha}(x):=\alpha(\hm(x),\hs_1(x),\hs_2(x))$. Then the following functions are Lipschitz on $[\ux,\ox]$
\begin{equation}
\begin{split}\label{lip}
\hat{\alpha}, \hat{\beta},\hat{\gamma},\hat{\theta}_1, \hat{\theta}_2, \tfrac{\hs_1(x)}{g'(x)},  \tfrac{\hs_2(x)}{g'(x)}, \tfrac{\hat{\beta}(x)}{g'(x)}.
\end{split}
\end{equation}
(3) For $x\in [\ux,\ox]$, we have
\begin{equation}
\begin{split}\label{envelop}
-\hat{\alpha}(x) g'(x) -(\hm(x)+\hat{\beta}(x)) + \hat{\gamma}(x)\big(\tfrac{x}{g'(x)}\big)' =0.
\end{split}
\end{equation}

\end{corollary}

\begin{proof} 
(1) \eqref{inf_ode} is a simple optimization of a quadratic function. By \proref{FBO} (4) and (5), we can see that the first order condition produces the minimizer as follows
\begin{equation}
\begin{split}\label{optimizer}
\hs_1(x)&=\tfrac{\sigma_1 ( x-q \,g(x)) g'(x)}{q \, g(x)(1+g'(x))- (1+q)x \, g'(x)},\\
\hs_2(x)&=\tfrac{\big(-\rho \sigma_1 \sigma_2 x - (\mu_2(1+q)^2-q \rho \sigma_1 \sigma_2)x g'(x) + q(1+q) \mu_2 g(x) (1+g'(x))     \big)   g'(x)}{\sigma_2\big(q \, g(x)(1+g'(x))- (1+q)x \, g'(x)\big)(1+g'(x))}, \\
\hm(x)&=\tfrac{1 }{2q(1+q) \sigma_2 g(x)} \Big(2(1-\rho^2) \sigma_2 (\sigma_1 +(1+q) \hs_1(x))(\sigma_1 + \hs_1(x)) x \\
&\quad+ q(1+q)\Big(2\mu_2 (\rho(\sigma_1+\hs_1(x))+\hs_2(x)) \\
&\quad-\sigma_2 \big(2\mu_1 + 2\sigma_1 \hs_1(x)+ \hs_1(x)^2 + 2\rho (\sigma_1+ \hs_1(x))\hs_2(x) + \hs_2(x)^2 \big) \Big) g(x)   \Big).
\end{split}
\end{equation}
Therefore, $(\hm(x),\hs_1(x),\hs_2(x))$ are well defined on $[\ux,\ox]$ because of \proref{FBO} (4). 

(2) The form of $\hs_1(x)$ and $\hs_2(x)$ above, and the observation
$$
\hs_1(x) + \sigma_1 =  \tfrac{q \sigma_1 (g(x)-x\, g'(x))}{q \, g(x)(1+g'(x))- (1+q)x \, g'(x)} \neq 0 \quad \textrm{for  } x\in [\ux,\ox],
$$
  show that the functions in \eqref{lip} are Lipschitz on $[\ux,\ox]$, by \proref{FBO} (4).

(3) The appropriate version of the Envelope
  Theorem (see, e.g., Theorem 3.3, p.~475 in \cite{GinKey02}) or direct computation produce \eqref{envelop}. 
\end{proof}

We construct the shadow price process using the solution $(g, \ox,\ux)$ of the free boundary problem in \proref{FBO}. As a preliminary, we define the functions $f,\xi,r:[\ux,\ox]\mapsto \bR$ as
\begin{equation}
\begin{split}\label{fr}
f(x):=& \uy+ \int_{x}^{\ox} \tfrac{g'(t)}{t} dt,\\
\xi(x):=& \eta_0 + \eta_1 S_0^{(1)} e^{f(x)}+\eta_2 S_0^{(2)},\\
r(x):=& \eta_1 S_0^{(1)} e^{f(x)} -  \xi(x)\tfrac{x}{q\, g(x)},
\end{split}
\end{equation}
where $\uy=\ln(1-\underline{\lambda})$ and $\oy=\ln(1+\overline{\lambda})$. Then $e^{f(\ux)}=(1+\overline{\lambda})$ and $e^{f(\ox)}=(1-\underline{\lambda})$. Let the constant $\hx\in [\ux,\ox]$ be defined by\\
\begin{equation}
\begin{split}\label{hx}
\hx =
\begin{cases}
\ox, & r(x)>0\text{ for all } x\in [\ux,\ox]\\
\ux, & r(x)<0\text{ for all } x\in [\ux,\ox]\\
\text{a solution to } r(x)=0, & otherwise.
\end{cases}
\end{split}
\end{equation}
Consider the following reflected (Skorokhod-type) SDE on the interval $[\ux,\ox]$: 
\begin{equation}
\begin{split}\label{X reflect}
\left\{   \begin{split}
dX_t &= \Big(X_t \hat{\alpha}(X_t) + \tfrac{X_t\hat{\beta}(X_t)}{g'(X_t)}\Big) dt - \tfrac{X_t \hs_1(X_t)}{g'(X_t)} dB_t^{(1)} - \tfrac{X_t \hs_2(X_t)}{g'(X_t)} dB_t^{(2)} + d\Phi_t \\
X_0&=\hx.
   \end{split}
\right.\\
\end{split}
\end{equation}
\corref{cor1} (2) implies that the coefficients of the above SDE are Lipschitz on $[\ux,\ox]$. Therefore, the classical result of \cite{Sko61} is applicable: \eqref{X reflect} has a unique solution $(X,\Phi)$ such that $\Phi$ is a continuous process of finite variation and satisfies
\begin{equation}
\begin{split}\label{reflect}
d\Phi_t^{\uparrow} = 1_{\{X_t= \ux\}} d\Phi_t^{\uparrow}, \quad d\Phi_t^{\downarrow} = 1_{\{X_t= \ox\}} d\Phi_t^{\downarrow}.
\end{split}
\end{equation}
We define the process (candidate shadow price process) $\tS$ as
\begin{equation}
\begin{split}\label{tS def}
\tS_t := S_t^{(1)} e^{f(X_t)}
\end{split}
\end{equation}
The intuition is the following. In Section 3, we change variable $(y,w)$ to $(x, g)$, and they satisfy $dy/dx = -g'(x)/x$ and $-w'(\uy)=\ox$, which implies $y=f(x)$. Also in Section 3, the shadow price process has the form of $S_t^{(1)} e^{Y_t}$.

\subsection{Verification argument} In this subsection, we verify that the process $\tS$ in \eqref{tS def} is indeed a shadow price process. First, we study properties of $\tS$.

\begin{proposition} 
(1) $ \uS_t^{(1)} \leq \tS_t \leq \oS_t^{(1)}$  for  $t\geq 0$ a.s.\\
(2) $\tS_t$ satisfies the SDE
\begin{equation}
\begin{split}\label{tS}
\frac{d\tS_t}{\tS_t} &= \Big( \hm(X_t)+\mu_1 + \sigma_1 \big(\hs_1(X_t)+\rho \hs_2(X_t)\big) + \hat{\gamma}(X_t) \Big) dt\\
&\qquad + (\hs_1(X_t)+\sigma_1)dB_t^{(1)} + \hs_2(X_t) dB_t^{(2)}
\end{split}
\end{equation}
\end{proposition}
\begin{proof}
(1) \proref{FBO} (5) implies that $f$ is a monotonically decreasing function. Hence $\uy \leq f(x) \leq \oy$, which implies $ \uS_t^{(1)} \leq \tS_t \leq \oS_t^{(1)}$. \\
(2) By Ito's formula,
\begin{equation}
\begin{split}\label{Y}
d(f(X_t))&=\Big(-\tfrac{g'(x)}{x} \Big(x \hat{\alpha}(x) + \tfrac{x \hat{\beta}(x)}{g'(x)}\Big) +\big(\tfrac{x}{g'(x)}\big)'  \hat{\gamma}(x)
\Big)\Big|_{x=X_t} dt \\
&\quad + \hs_1(X_t) dB_t^{(1)} + \hs_2(X_t) dB_t^{(2)} - \tfrac{g'(x)}{x}  d\Phi_t \\
&=\hm(X_t) dt + \hs_1(X_t) dB_t^{(1)} + \hs_2(X_t) dB_t^{(2)},
\end{split}
\end{equation}
where the $dt$ term is simplified by \eqref{envelop}, and the reflection term $( \tfrac{g'(x)}{x}  d\Phi_t)$ vanishes because of $g'(\ux)=g'(\ox)=0$ and \eqref{reflect}. Ito's formula for $\tS_t = S_t^{(1)} e^{f(X_t)}$, together with \eqref{S} and \eqref{Y}, produces \eqref{tS}.
\end{proof}

In the frictionless market with $(\tS,S^{(2)})$,  the state price density process $\hH$ is given by
\begin{equation}
\begin{split}\label{stateprice}
\frac{d\hH_t}{\hH_t} = -\hat{\theta}_1(X_t) dB_t^{(1)} - \hat{\theta}_2(X_t) dB_t^{(2)}, \quad \hH_0=1.
\end{split}
\end{equation}

Consider the optimization problem in the frictionless market $(\tS,S^{(2)})$
\begin{equation}
\begin{split}\label{shadow problem}
\sup_{c\in \sC(\tS)}  \bE\Big[\int_0^\infty e^{-\delta t} U(c_t)dt \Big].
\end{split}
\end{equation}

In the next lemma, we characterize the value and the optimal strategy for \eqref{shadow problem}.

\begin{lemma}\label{shadow solve}
(1) Let $\tS$ and $\hH$ be as in \eqref{tS def} and \eqref{stateprice}. Then
\begin{equation}
\begin{split}\label{shadow value}
\sup_{c\in \sC(\tS)}  \bE\Big[\int_0^\infty e^{-\delta t} U(c_t)dt \Big] 
=\tfrac{\xi(\hx)^p}{p} |g(\hx)|^{1-p}.
\end{split}
\end{equation}
(2) In \eqref{shadow value}, the optimal wealth $\hW$ and the optimal consumption/investment $(\hc,\hat{\varphi}^{(0)},\hat{\varphi}^{(1)},\hat{\varphi}^{(2)})$ can be written as following
\begin{equation}
\begin{split}\label{strategy}
\hW_t &=\xi(\hx) e^{-(1+q)\delta t} \hH_t^{-(1+q)}\frac{g(X_t)}{g(\hx)},\\
\hc_t &= \tfrac{\hW_t}{|g(X_t)|}, \,\,\, \hat{\varphi}^{(0)}_t = (1-\pi_1(X_t)-\pi_2(X_t))\hW_t,\,\,\,\hat{\varphi}^{(1)}_t = \tfrac{\pi_1(X_t) \hW_t}{\tS_t},\,\,\, \hat{\varphi}^{(2)}_t = \tfrac{\pi_2(X_t) \hW_t}{S_t^{(2)}},\,\,\,
\end{split}
\end{equation}
where the functions $\pi_1,\pi_2:[\ux,\ox]\mapsto \mathbb{R}$ are
\begin{equation}
\begin{split}\label{pi12}
\pi_1(x)&:=\tfrac{(1+q)\hat{\theta}_1(x)-\frac{x \hs_1(x)}{g(x)} }{\hs_1(x)+\sigma_1}, \\
 \pi_2(x)&:= \tfrac{1}{\sigma_2}\Big((1+q) \hat{\theta}_2(x) g(x)- \tfrac{x \hs_2(x)}{g(x)} - \pi_1(x) \hs_2(x)\Big).
\end{split}
\end{equation}
\end{lemma}
\begin{proof}
(1) We first prove the following equality
\begin{equation}
\begin{split}\label{g rep}
g(\hx)= \sgn(p) \bE\Big[\int_0^\infty e^{-(1+q)\delta t} \hH_t^{-q}dt \Big].
\end{split}
\end{equation}
Using Ito's formula and \eqref{envelop}, we have
\begin{equation}
\begin{split}\label{dg}
d(g(X_t))= \big(-X_t \hm(X_t) + \tfrac{X_t \hat{\gamma}(X_t)}{g'(X_t)}\big)dt - X_t \hs_1(X_t)dB_t^{(1)}- X_t \hs_2(X_t)dB_t^{(2)},
\end{split}
\end{equation}
where the reflection term vanishes because of $g'(\ux)=g'(\ox)=0$ and \eqref{reflect}.

Observe that the stochastic exponential $\sE (q \hat{\theta} \cdot B)$, with $\hat{\theta}_t = (\hat{\theta}_1(X_t),\hat{\theta}_2(X_t))$ and $B_t=(B_t^{(1)},B_t^{(2)})$, is a martingale because $\hat{\theta}$ is bounded. Let $\bar{B}^{(1)},\bar{B}^{(2)}$ be defined by
\begin{displaymath}
\begin{split}
\bar{B}^{(1)}_t:= B_t^{(1)} - q \int_0^t \hat{\theta}_1(X_s) + \rho \, \hat{\theta}_2(X_s) \, ds,\quad
\bar{B}^{(2)}_t:= B_t^{(2)} - q \int_0^t \rho \, \hat{\theta}_1(X_s) + \hat{\theta}_2(X_s) \, ds.
\end{split}
\end{displaymath}
As $\hat{\theta}_1$ and $\hat{\theta}_2$ are bounded on $[\ux,\ox]$, by Girsanov's theorem, $\bar{B}^{(1)}$ and $\bar{B}^{(2)}$ are Brownian motions on $[0,t]$ under the measure $\bar{\bP}_t$, defined by $d\bar{\bP}_t= \sE (q \hat{\theta} \cdot B)_t \,d\bP$. Then,
\begin{equation}
\begin{split}\label{veri}
&\bE^{\bar{\bP}_t}\Big[ e^{-\int_0^t \hat{\alpha}(X_u)du}g(X_t) \Big]\\
&=g(\hx) -  \bE^{\bar{\bP}_t}\Big[ \int_0^t e^{-\int_0^u \hat{\alpha}(X_s) ds} \Big(\sgn(p)+ X_u \big(\hs_1(X_u) d\bar{B}_u^{(1)} +\hs_2(X_u) d\bar{B}_u^{(2)} \big)     \Big) du  \Big] \\
&=g(\hx) -\sgn(p)\,  \bE^{\bar{\bP}_t}\Big[\int_0^t e^{-\int_0^u \hat{\alpha}(X_s) ds} du\Big] \\
%&=g(\hx) -\sgn(p) \int_0^t \bE \Big[ e^{-(1+q)\delta u} \hH_u^{-q} \Big] du\\
&=g(\hx)-\sgn(p)\,  \bE \Big[\int_0^t e^{-(1+q)\delta u} \hH_u^{-q} du\Big]
\end{split}
\end{equation}
Here the first equality uses Ito's formula and \eqref{inf_ode}, and the second equality holds because $B_t^{(1)}$ and $B_t^{(2)}$ are Brownian motions under the measure $\bar{\bP}_t$ and the integrands are bounded. The third equality is due to \eqref{stateprice} and $d\bar{\bP}_t= \sE (q \hat{\theta} \cdot B)_t d\bP$.

We have two cases to consider, $p>0$ and $p<0$.\\
(i) In case $p>0$: Because $g(x)$ is positive (see \proref{FBO} (4)), \eqref{veri} implies that $$\bE [\int_0^{\infty} e^{-(1+q)\delta t} \hH_t^{-q} dt]<\infty.$$ Hence, there exists a sequence $(t_n)_{n\in \bN}$ with $t_n \to \infty$ such that $\bE[e^{-(1+q)\delta t_n} \hH_{t_n}^{-q}] \to 0$. Because $g$ is bounded, we have 
$$\bE^{\bar{\bP}_{t_n}}\Big[ e^{-\int_0^{t_n} \hat{\alpha}(X_u)du}g(X_{t_n}) \Big]=\bE[e^{-(1+q)\delta t_n} \hH_{t_n}^{-q}g(X_{t_n})] \to 0 \quad \textrm{as  } t_n \to \infty $$
Therefore, we take limit $t_n\to \infty$ in \eqref{veri} and conclude \eqref{g rep}. \\
(ii) In case $p<0$: From the form of the function $\alpha$ in \eqref{ab} and $q<0$, we have $\hat{\alpha}> (1+q)\delta$.  Because $g$ is bounded, 
\begin{equation}
\begin{split}
\Big|\bE^{\bar{\bP}_t}\Big[ e^{-\int_0^t \hat{\alpha}(X_u)du}g(X_t) \Big]  \Big| \leq |g|_{\infty} e^{-(1+q)\delta t}  \to 0 \quad \textrm{as  }t\to \infty.
\end{split}
\end{equation}
Let $t \to \infty$ in \eqref{veri}, we conclude \eqref{g rep}.

Therefore, we conclude that \eqref{g rep} holds for all cases. Now the standard duality theory for complete market model (see, e.g., Theorem 9.11, p.~141 in \cite{KarShr98})\footnote{For the application of Theorem 9.11 in \cite{KarShr98}, one needs to check Assumption 9.9 in \cite{KarShr98}. In the current setup, Assumption 9.9 amounts to $\bE [\int_0^{\infty} e^{-(1+q)\delta t} \hH_t^{-q} dt]<\infty$, which is true by \eqref{g rep}. } implies that
\begin{equation}
\begin{split}\label{duality thm}
\sup_{c\in \sC(\tS)}  \bE\Big[\int_0^\infty e^{-\delta t} U(c_t)dt \Big] =\tfrac{ \xi(\hx)^p}{p} \Big(\bE\Big[\int_0^\infty e^{-(1+q)\delta t} \hH_t^{-q}dt \Big]  \Big)^{1-p},
\end{split}
\end{equation}
and the optimal consumption $\hc$ is
\begin{equation}
\begin{split}\label{c raw}
\hc_t=\frac{\xi(\hx) e^{-(1+q)\delta t} \hH_t^{-(1+q)}}{\bE\big[\int_0^\infty e^{-(1+q)\delta t} \hH_t^{-q}dt \big]}.
\end{split}
\end{equation}
By \eqref{g rep} and \eqref{duality thm}, we obtain \eqref{shadow value}.

\smallskip

(2) Obviously, $\hW_t>0$ for $t\geq 0$. As we have \eqref{c raw} and \eqref{g rep}, it is enough to check the budget constraint in \defref{consistent} (2). It can be written as
\begin{equation}
\begin{split}\label{dW}
\frac{d\hW_t}{\hW_t} = \pi_1(X_t) \frac{d\tS_t}{\tS_t} + \pi_2(X_t) \frac{dS_t^{(2)}}{S_t^{(2)}} -\frac{\hc_t}{\hW_t} dt.
\end{split}
\end{equation}
Using Ito's formula with \eqref{stateprice}, \eqref{dg}, \eqref{tS}, \eqref{pi12} and \eqref{inf_ode}, one can check that the budget constraint holds (the computation is rather long and tedious but elementary, so it is omitted).
\end{proof}

Now we are ready to state our main result. In \thmref{main verify}, we verify that the process $\tS$ in \eqref{tS def} is indeed a shadow price process. Consequently, the optimal consumption/trading strategy $(\hc,\hat{\varphi}^{(0)},\hat{\varphi}^{(1)},\hat{\varphi}^{(2)})$ of the frictionless problem \eqref{shadow problem} also satisfies \eqref{adm} and \eqref{budget}, and $\hc\in \sC$ is the optimizer of \eqref{primal problem}.

\begin{theorem}\label{main verify}(Existence of the shadow price) Under the assumptions in \proref{FBO}, the processes $(\hc,\hat{\varphi}^{(0)},\hat{\varphi}^{(1)},\hat{\varphi}^{(2)})$ in \eqref{strategy} solves \eqref{primal problem}. In other words, $(\hc,\hat{\varphi}^{(0)},\hat{\varphi}^{(1)},\hat{\varphi}^{(2)})$ satisfies the conditions in \defref{financeable}  (therefore $\hc\in \sC$), and
\begin{displaymath}
\begin{split}
\sup_{c\in \sC} \bE\Big[\int_0^\infty e^{-\delta t} U(c_t)dt \Big]&=\bE\Big[\int_0^\infty e^{-\delta t} U(\hc_t)dt \Big]=\tfrac{ \xi(\hx)^p}{p} |g(\hx)|^{1-p}.
\end{split}
\end{displaymath}
Indeed, $\tS$ is a shadow price process.
\end{theorem}
\begin{proof}
By \lemref{shadow solve}, we already know that $(\hc,\hat{\varphi}^{(0)},\hat{\varphi}^{(1)},\hat{\varphi}^{(2)})$ is the optimal solution of \eqref{shadow pro}. Therefore, we only need to check that $(\hc,\hat{\varphi}^{(0)},\hat{\varphi}^{(1)},\hat{\varphi}^{(2)})$ satisfies the assumptions in \proref{shadow_connection} (2). Then, the result of \proref{shadow_connection} completes the proof of this theorem.  

Let's first consider the initial jump. We need to show that the assumption (iii) in \proref{shadow_connection} (2) is satisfied at $t=0$, which can be written as
\begin{equation}
\begin{split}\label{initial jump2}
\hat{\varphi}^{(1)}_0 - \eta_1=1_{\{ \hx=\ux \}}(\hat{\varphi}^{(1)}_0 - \eta_1)^+ - 1_{\{ \hx=\ox \}}(\hat{\varphi}^{(1)}_0 - \eta_1)^-.
\end{split}
\end{equation}

In \eqref{pi12}, we can simplify $\pi_1(x)$ as $\pi_1(x)=\frac{x}{q g(x)}$ by using expressions in \eqref{optimizer}. Then $r(x)$ in \eqref{fr} can be written as $r(x)=\eta_1 e^{f(x)} S_0^{(1)} - \xi(x)\pi_1(x)$. Now we can see why we defined $\hx\in [\ux,\ox]$ as \eqref{hx}. The three possibilities are
\begin{equation}
\begin{split}\label{r cases}
\begin{cases}
\hat{\varphi}_0^{(1)}= \eta_1, &\textrm{if $r(\hx)=0$},\\
 \textrm{$\hat{\varphi}_0^{(1)}< \eta_1$ and $\hx=\ox$}, & \textrm{if $r(x)>0$ on $[\ux,\ox]$,}\\
 \textrm{$\hat{\varphi}_0^{(1)}> \eta_1$ and $\hx=\ux$}, &\textrm{if $r(x)<0$ on $[\ux,\ox]$.}
\end{cases}
\end{split}
\end{equation}
Obviously \eqref{r cases} implies \eqref{initial jump2}, and we conclude that assumption (iii) in \proref{shadow_connection} (2) is satisfied at $t=0$.

By \proref{FBO} and the form of $\pi_1$, we observe that $\hat{\varphi}^{(1)}_t>0$ if $\mu_1>\frac{\rho \mu_2 \sigma_1}{\sigma_2}$ and $\hat{\varphi}^{(1)}_t<0$ if $\mu_1<\frac{\rho \mu_2 \sigma_1}{\sigma_2}$. With \eqref{tS}, \eqref{X reflect}, \eqref{dW} and \eqref{dg}, Ito's formula produces (after a long but straightforward computation)
\begin{equation}
\begin{split}
\begin{cases}\label{ln}
d(\ln(\hat{\varphi}^{(1)}_t))=d\Big( \ln\Big(\frac{\pi_1(X_t) \hW_t}{\tS_t}\Big)  \Big) =\frac{1}{X_t} d\Phi_t, &\textrm{when  } \mu_1>\frac{\rho \mu_2 \sigma_1}{\sigma_2} , \\
d(\ln(-\hat{\varphi}^{(1)}_t))=d\Big( \ln\Big( \frac{-\pi_1(X_t) \hW_t}{\tS_t} \Big) \Big) =\frac{1}{X_t} d\Phi_t,&\textrm{when  } \mu_1<\frac{\rho \mu_2 \sigma_1}{\sigma_2}.
\end{cases}
\end{split}
\end{equation}
\eqref{ln} and \eqref{reflect} implies that the assumptions (i) and (iii) in \proref{shadow_connection} (2) are satisfied.

Because assumption (iv) is obvious, it remains to check assumption (ii) in \proref{shadow_connection} (2). This amounts to proving that
\begin{equation}
\begin{split}\label{adm check}
\hat{\varphi}_t^{(0)}+\underline{S}_t^{(1)} (\hat{\varphi}_t^{(1)})^+ -\overline{S}_t^{(1)} (\hat{\varphi}_t^{(1)})^- + S_t^{(2)}\hat{\varphi}_t^{(2)} \geq 0, \quad t\geq 0.
\end{split}
\end{equation}
Using \proref{FBO} (4) and (5), we obtain following inequalities
\begin{equation}
\begin{split}\label{pi monotone}
\tfrac{d}{dx}\Big( \tfrac{\pi_1(x) \, e^{-f(x)}}{1-\pi_1(x)}   \Big) &=\tfrac{ \big(q \,  g(x) \left(g'(x)+1\right)-(1+q) x g'(x)\big) e^{-f(x)}}{q^2 g(x)^2 (1-\pi_1(x)^2} >0, \quad x\in[\ux,\ox]\\
\tfrac{d}{dx}\pi_1(x)&=\tfrac{q(g(x)-x g'(x))}{q^2 g(x)^2}>0, \quad x\in[\ux,\ox]
\end{split}
\end{equation}
$\bullet$ In case $\mu_1>\frac{\rho \mu_2 \sigma_1}{\sigma_2}$: By \proref{FBO} (1) and (3), we have $\pi_1(x)>0$, so $\hat{\varphi}_t^{(1)}>0$. \\
If $\pi_1(X_t)\leq1$, then $\hat{\varphi}_t^{(0)}+\hat{\varphi}_t^{(2)}S_t^{(2)} = (1-\pi_1(X_t))\hW_t \geq 0$. Therefore, \eqref{adm check} holds. \\
If $\pi_1(X_t)>1$, then $\hat{\varphi}_t^{(0)}+\hat{\varphi}_t^{(2)}S_t^{(2)} <0$. \eqref{pi monotone} implies that 
  $$\tfrac{\hat{\varphi}_t^{(0)}+\underline{S}_t^{(1)} \hat{\varphi}_t^{(1)} +S_t^{(2)}\hat{\varphi}_t^{(2)}}{\hat{\varphi}_t^{(0)} + S_t^{(2)}\hat{\varphi}_t^{(2)}}=(1-\underline{\lambda})\tfrac{\pi_1(X_t)\, e^{-f(X_t)}}{1-\pi_1(X_t)}  +1\leq (1-\underline{\lambda})\tfrac{\pi_1(\ox)\, e^{-f(\ox)}}{1-\pi_1(\ox)} +1=\tfrac{1}{1-\pi_1(\ox)}<0,$$
where we use $e^{-f(\ox)}=1/(1-\underline{\lambda})$. Hence \eqref{adm check} holds.\\
$\bullet$ In case $\mu_1<\frac{\rho \mu_2 \sigma_1}{\sigma_2}$: By \proref{FBO} (1) and (3), we have $\pi_1(x)<0$, so $\hat{\varphi}_t^{(1)}<0$ and $ \hat{\varphi}_t^{(0)}+\hat{\varphi}_t^{(2)}S_t^{(2)} = (1-\pi_1(X_t))\hW_t>0$. \eqref{pi monotone} implies that 
  $$\tfrac{\hat{\varphi}_t^{(0)}+\overline{S}_t^{(1)} \hat{\varphi}_t^{(1)} +S_t^{(2)}\hat{\varphi}_t^{(2)}}{\hat{\varphi}_t^{(0)} + S_t^{(2)}\hat{\varphi}_t^{(2)}}=(1+\overline{\lambda})\tfrac{\pi_1(X_t)\, e^{-f(X_t)} }{1-\pi_1(X_t)} +1\geq (1+\overline{\lambda})\tfrac{\pi_1(\ux)\,  e^{-f(\ux)}}{1-\pi_1(\ux)} +1=\tfrac{1}{1-\pi_1(\ux)}>0,$$
where we use $e^{-f(\ux)}=1/(1+\overline{\lambda})$. Hence \eqref{adm check} holds.

We showed that $(\hc,\hat{\varphi}^{(0)},\hat{\varphi}^{(1)},\hat{\varphi}^{(2)})$ satisfies the assumptions in \proref{shadow_connection} (2), and the proof is completed by the result of \proref{shadow_connection}.
\end{proof}

\begin{remark}
\thmref{main verify} does not rely on a \it{smallness assumption} for the transaction cost parameters $\overline{\lambda}$ and $\underline{\lambda}$.
\end{remark}

\subsection{Well-posedness of the problem}

The result in the previous subsection enables us to explicitly characterize when the optimal consumption and investment problem \eqref{primal problem} is well-posed (i.e., the value is finite). Recall that \assref{ass} is the standing assumption in this paper, and if \assref{ass} (1) is violated, then the problem is ill-posed (see \remref{finite value}).

\begin{theorem}\label{well-posed thm} (Well-posedness of the problem) The following statements are equivalent.

(1) The optimization problem \eqref{primal problem} is well-posed, i.e.,
$$
\sup_{c\in \sC} \bE\Big[\int_0^\infty e^{-\delta t} U(c_t)dt \Big]<\infty.
$$

(2) There exists a shadow price process.

(3) The model parameters satisfy one of the following three conditions:
\begin{itemize}
\item[(i)] $p \leq 0$,
\item[(ii)] $0<p<1$ and $\delta > \tfrac{q}{2(1-\rho^2)} \big( (\tfrac{\mu_1}{\sigma_1})^2+ (\tfrac{\mu_2}{\sigma_2})^2- 2\rho \tfrac{\mu_1 \mu_2}{\sigma_1 \sigma_2} \big) $,
\item[(iii)] $0<p<1$, $\delta \leq  \tfrac{q}{2(1-\rho^2)} \big( (\tfrac{\mu_1}{\sigma_1})^2+ (\tfrac{\mu_2}{\sigma_2})^2- 2\rho \tfrac{\mu_1 \mu_2}{\sigma_1 \sigma_2} \big) $ and 
$c^*< \ln \big(\tfrac{1+\overline{\lambda}}{1-\underline{\lambda}}\big)$,
\item[]   where $c^*$ is a constant explicitly defined in \defref{c definition}. 
\end{itemize}
\end{theorem}

\begin{proof}
See Section 6.
\end{proof}

\begin{remark}
\thmref{well-posed thm} provides an explicit characterization of the well-posedness of the problem, in the sense that the constant $c^*$ is given by a closed form in terms of the model parameters. \cite{Hobson16} also provides a well-posedness criterion by analyzing the HJB equation of the primal optimization problem.
\end{remark}

\section{Discussion of optimal strategy}

The expression  \eqref{strategy} enables us to extract more information about how the transaction costs affect the optimal consumption/investment strategy. For convenience, in this section, we set $\overline{\lambda}=0$ and $\underline{\lambda}=\lambda$, and only consider the specific case of \assref{ass5}. This assumption means that the proportion of wealth invested in the illiquid asset should be positive. Other cases can be analyzed similarly.

\begin{assumption}\label{ass5}
In this section, we assume that the following inequalities hold.
\begin{displaymath}
p>0, \quad \mu_1>\tfrac{\rho \mu_2 \sigma_1}{\sigma_2}, \quad \delta>\tfrac{q}{2(1-\rho^2)} ( (\tfrac{\mu_1}{\sigma_1})^2+ (\tfrac{\mu_2}{\sigma_2})^2- 2\rho \tfrac{\mu_1 \mu_2}{\sigma_1 \sigma_2} ).
\end{displaymath}
\end{assumption}

To observe the effect of the transaction costs, we remind the case of $\lambda=0$, which is the classical Merton problem. When $\lambda=0$, it is well known that (cf. \cite{Mer71}) the optimal proportion of the consumption rate and investment are given by
\begin{equation}
\begin{split}\label{MertonOpt}
c^M &:= (1+q)\Big(\delta -\tfrac{q}{2(1-\rho^2)} ( (\tfrac{\mu_1}{\sigma_1})^2+ (\tfrac{\mu_2}{\sigma_2})^2- 2\rho \tfrac{\mu_1 \mu_2}{\sigma_1 \sigma_2} )\Big),\\
\pi_1^M &:= \tfrac{(1+q)(\mu_1-\frac{\rho \sigma_1}{\sigma_2} \mu_2 ) }{(1-\rho^2)\sigma_1^2}, \quad \pi_2^M := \tfrac{(1+q)(\mu_2-\frac{\rho \sigma_2}{\sigma_1} \mu_1 ) }{(1-\rho^2)\sigma_2^2}.
\end{split}
\end{equation}

%The following observations describe some properties (and intuitive explanations) of the optimal strategy for illiquid asset trading, liquid risky asset trading, and consumption. Mathematics for the observations is sketched in Appendix.

%For the further analysis of the effect of transaction costs on the optimal investment and consumption strategy, we provide asymptotic result for small transaction costs, as in \cite{Cho11}.

\begin{proposition}\label{asymptotic_seed}
(Under \assref{ass5})  $\ux$, $\ox$, $g(\ux)$, and $g(\ox)$ in \proref{FBO}  have the following expansions for small transaction cost $\lambda$ 
\begin{equation}
\begin{split}\label{asymptotic}
\ux &= \tfrac{q \, \pi_1^M}{c^M} -\tfrac{q \, \zeta}{c^M} \lambda^{\frac{1}{3}} +  O(\lambda^{\frac{2}{3}}) , \quad \\
\ox &= \tfrac{q \, \pi_1^M}{c^M} + \tfrac{q \, \zeta}{c^M}\lambda^{\frac{1}{3}} + O(\lambda^{\frac{2}{3}}), \\
g(\ux)&=\tfrac{1}{c^M}  - \tfrac{q(1-\rho^2)\sigma_1^2  \zeta^2}{2 (1+q) (c^M)^2 }  \lambda^{\frac{2}{3}} + O(\lambda), \quad\\
g(\ox)&=\tfrac{1}{c^M}  -\tfrac{q(1-\rho^2)\sigma_1^2  \zeta^2}{2 (1+q) (c^M)^2 }   \lambda^{\frac{2}{3}} + O(\lambda), \\
\end{split}
\end{equation}
where 
\begin{equation}\label{zeta_exp}
\zeta:=\Big( \tfrac{3(1+q) (\pi_1^M)^2 (1-\pi_1^M)^2}{4} + \tfrac{3(1+q)(\mu_2(1+q)-\rho \sigma_1 \sigma_2)^2 (\pi_1^M)^2}{4(1-\rho^2)\sigma_1^2 \sigma_2^2 }   \Big)^\frac{1}{3}.
\end{equation}
\end{proposition}
\begin{proof}
We apply the methodology in \cite{Cho11} (for the case of one illiquid asset only) to obtain the expansions. The computations are straightforward but tedious, so we omit the details here.
\end{proof}

\subsection{Optimal trading of the illiquid asset}

 We first find a more explicit characterization for the optimal investment in the illiquid asset.

\begin{corollary}\label{no trading verify}
(Under \assref{ass5}) In \eqref{primal problem}, it is optimal to minimally trade the illiquid asset $S^{(1)}$ in such a way that the proportion of investment in the illiquid asset is within the interval $[\underline{\pi}_1,\overline{\pi}_1]$, i.e.,
\begin{equation}
\begin{split}
\underline{\pi}_1 \leq \tfrac{ \hat{\varphi}^{(1)}_t S_t^{(1)} }{ \hat{\varphi}^{(0)}_t+\hat{\varphi}^{(1)}_t S_t^{(1)}+\hat{\varphi}^{(2)}_t S_t^{(2)}} \leq \overline{\pi}_1,
\end{split}
\end{equation}
where $\underline{\pi}_1,\overline{\pi}_1 \in \bR$ have explicit expressions in terms of $g,\ux,\ox$ in \proref{FBO},
\begin{equation}
\begin{split}\label{pidef}
\underline{\pi}_1:=\pi_1 (\ux),\quad \overline{\pi}_1:=\tfrac{\pi_1 (\ox)}{\pi_1(\ox) + (1-\lambda)(1-\pi_1(\ox))}
\end{split}
\end{equation}
\end{corollary}
\begin{proof}
We can easily transform 
$$\pi_1(X_t)=\tfrac{ \hat{\varphi}^{(1)}_t \tS_t }{ \hat{\varphi}^{(0)}_t+\hat{\varphi}^{(1)}_t \tS_t+\hat{\varphi}^{(2)}_t S_t^{(2)}} \quad \Longrightarrow \quad
\tfrac{ \hat{\varphi}^{(1)}_t S_t^{(1)} }{ \hat{\varphi}^{(0)}_t+\hat{\varphi}^{(1)}_t S_t^{(1)}+\hat{\varphi}^{(2)}_t S_t^{(2)}}=\tfrac{\pi_1(X_t)}{\pi_1(X_t) + (1- \pi_1(X_t))e^{f(X_t)}}. $$
Direct computation produces the following inequality
$$\tfrac{d}{dx}\Big(\tfrac{\pi_1(x)}{\pi_1(x) + (1- \pi_1(x))e^{f(x)}}\Big)= \tfrac{ \big(q \,  g(x) \left(g'(x)+1\right)-(1+q) x g'(x)\big) \,e^{f(x)}}{q^2 g(x)^2 ((e^{f(x)}-1)\pi_1(x)- e^{f(x)})^2}>0, \quad x\in [\ux,\ox],$$
where we use the result in \proref{FBO} (4). Therefore, we have
$$\pi_1 (\ux)\leq \tfrac{\pi_1(X_t)}{\pi_1(X_t) + (1- \pi_1(X_t))e^{f(X_t)}}\leq \tfrac{\pi_1 (\ox)}{\pi_1(\ox) + (1-{\lambda})(1-\pi_1(\ox))}, \quad t \geq 0,$$
and the result follows.
\end{proof}

%\begin{remark}
%As we pointed out in \remref{special cases}, we can describe results similar to \thmref{main verify} and \corref{no trading verify} for the case of  $\mu_1 = \frac{\rho \mu_2 \sigma_1}{\sigma_2}$ or  $\mu_2 = \frac{\rho \sigma_1 \sigma_2}{1+q}$, by using the results in \cite{ChoSirZit11}.
%\end{remark}

\begin{corollary}\label{pi1_cor}
(Under \assref{ass5})  $\underline{\pi}_1$ and $\overline{\pi}_1$ in \eqref{pidef} have the following expansions for small transaction cost $\lambda$ 
\begin{equation}
\begin{split}\label{asymptotic_pi1}
\underline{\pi}_1 &= \pi_1^{M} -\zeta \lambda^{\frac{1}{3}} +  O(\lambda^{\frac{2}{3}}) , \quad \\
\overline{\pi}_1 &= \pi_1^{M} + \zeta \lambda^{\frac{1}{3}} +  O(\lambda^{\frac{2}{3}}).
\end{split}
\end{equation}
\end{corollary}
\begin{proof}
The expression of $\pi_1(x)$ in \eqref{pi12} can be rewritten as 
\begin{align}\label{simplie_pi1}
\pi_1(x)=\tfrac{x}{q\, g(x)}.
\end{align}
The above expression, together with \eqref{pidef} and \proref{asymptotic_seed}, produces \eqref{asymptotic_pi1}.
\end{proof}

\corref{no trading verify} and \corref{pi1_cor} have the following implications regarding the optimal trading of the illiquid asset. Assume that the transaction costs $\lambda>0$ is small enough.
 The no-trading region described in \corref{no trading verify} is {\it wider} than the no-trading region in the model without the liquid risky asset.\footnote{This phenomenon is also observed in \cite{Guasoni} for the case of $\rho=0$.}
 Indeed, the model {\it without} the liquid risky asset is studied in \cite{Cho11} and the width of the no-trading region is approximately
$$
2\Big( \tfrac{3(1+q) (\pi^M)^2 (1-\pi^M)^2}{4}  \Big)^\frac{1}{3}\lambda^{\frac{1}{3}},
$$
where $\pi^M$ (Merton proportion) is the proportion of wealth invested in the illiquid asset. In our model {\it with} the liquid risky asset, \corref{pi1_cor} implies that the width of the no-trading region is approximately
$$
2\Big( \tfrac{3(1+q) (\pi_1^M)^2 (1-\pi_1^M)^2}{4} + \tfrac{3(1+q)(\mu_2(1+q)-\rho \sigma_1 \sigma_2)^2 (\pi_1^M)^2}{4(1-\rho^2)\sigma_1^2 \sigma_2^2 }   \Big)^\frac{1}{3}\lambda^{\frac{1}{3}},
$$
which is bigger than the previous width, as long as the Merton proportions in two models agree, i.e., $\pi^M=\pi_1^M$. 

One possible explanation for this effect is following. Intuitively, the no-trading wedge is determined to minimize 
$$\qquad\qquad \textrm{(trading costs due to rebalancing) $+$ (reduction of the value for not rebalancing).} $$
The additional opportunity of investment in the liquid risky asset increases the volatility of the total wealth, and induces more trading costs due to more frequent rebalancing. Consequently, the no-trading region becomes wider to mitigate the trading costs. 

\subsection{Optimal trading of the liquid asset} We now explore how the transaction costs affect the trading of the liquid risky asset.

\begin{corollary}\label{pi2_cor}
(Under \assref{ass5})  In \eqref{primal problem}, the optimal proportion of wealth invested in the liquid risky asset has the following expansions for small transaction cost $\lambda$: 
\begin{equation}
\begin{split}\label{asymptotic_pi2}
\pi_2^M - \tfrac{\rho \sigma_1 \zeta}{\sigma_2}   \lambda^{\frac{1}{3}} +  O(\lambda^{\frac{2}{3}})  \quad & \textrm{when selling the illiquid asset,} \\
\pi_2^M + \tfrac{\rho \sigma_1 \zeta}{\sigma_2}   \lambda^{\frac{1}{3}} +  O(\lambda^{\frac{2}{3}})  \quad & \textrm{when buying the illiquid asset.} \\
\end{split}
\end{equation}
\end{corollary}
\begin{proof}
The expression of $\pi_2(x)$ in \eqref{pi12} can be rewritten as 
$$
\pi_2(x)=\tfrac{ (1+q)\mu_2}{\sigma_2^2} - \pi_1(x)\Big(  \tfrac{\rho \sigma_1}{\sigma_2} + \tfrac{(1+q)((1+q)\mu_2 - \rho \sigma_1 \sigma_2) g'(x)}{\sigma_2^2(1+g'(x))} \Big).
$$
The above expression, together with \proref{asymptotic_seed} and \corref{pi1_cor}, produces 
\begin{equation}
\begin{split}
\textrm{(liquid risky asset proportion)} 
&= \tfrac{ \hat{\varphi}^{(2)}_t S_t^{(2)} }{ \hat{\varphi}^{(0)}_t+\hat{\varphi}^{(1)}_t S_t^{(1)}+\hat{\varphi}^{(2)}_t S_t^{(2)}}\\
&=\tfrac{e^{f(X_t)}\pi_2(X_t)}{\pi_1(X_t) + (1- \pi_1(X_t))e^{f(X_t)}}\\
&=\begin{cases}
\pi_2^M - \tfrac{\rho \sigma_1 \zeta}{\sigma_2}   \lambda^{\frac{1}{3}} +  O(\lambda^{\frac{2}{3}}) , & \textrm{when  } X_t=\overline{x}\\
\pi_2^M + \tfrac{\rho \sigma_1 \zeta}{\sigma_2}   \lambda^{\frac{1}{3}} +  O(\lambda^{\frac{2}{3}}) , & \textrm{when  } X_t=\underline{x}\\
\end{cases}.
\end{split}
\end{equation}
Because it is optimal to sell (resp., buy) the illiquid asset when $X_t=\ox$ (resp., $X_t=\ux$), we conclude \eqref{asymptotic_pi2}.
\end{proof}

\corref{pi2_cor} has the following implications regarding the optimal trading of the liquid risky asset. Assume that the transaction costs $\lambda>0$ is small enough.
We compare the liquid risky asset trading strategies in our model and in the model without transaction costs ($\lambda=0$) in \eqref{MertonOpt}. This will show how the existence of the transaction costs on one asset may affect trading strategy of the other asset (still liquid). \corref{pi2_cor} implies that
\begin{itemize}
\item When $\rho>0$ and  the illiquid asset proportion $ \tfrac{ \hat{\varphi}^{(1)}_t S_t^{(1)} }{ \hat{\varphi}^{(0)}_t+\hat{\varphi}^{(1)}_t S_t^{(1)}+\hat{\varphi}^{(2)}_t S_t^{(2)}} $ is close to $\overline{\pi}_1$ (resp., $\underline{\pi}_1$), the liquid risky asset proportion $ \tfrac{ \hat{\varphi}^{(2)}_t S_t^{(2)} }{ \hat{\varphi}^{(0)}_t+\hat{\varphi}^{(1)}_t S_t^{(1)}+\hat{\varphi}^{(2)}_t S_t^{(2)}} $ is smaller (resp., bigger) than $\pi_2^M$. \\
\item When $\rho<0$ and  the illiquid risky asset proportion is close to $\overline{\pi}_1$ (resp., $\underline{\pi}_1$), the liquid risky asset proportion is bigger (resp., smaller) than $\pi_2^M$. 
\end{itemize}
The intuition for this effect is the following. Consider the case of $\rho>0$. When the illiquid asset proportion is larger than the desired proportion, the investor is overexposed to the risk factor of the illiquid asset. Because $\rho>0$, by reducing the proportion of the liquid risky asset, the exposure to the risk factor of the illiquid asset can be reduced accordingly. A similar explanation can be applied to the other cases.

\subsection{Optimal consumption rate} Finally, we examine the effect of the transaction costs on the optimal consumption rate.

\begin{corollary}\label{c_cor}
(Under \assref{ass5}) In \eqref{primal problem}, the optimal consumption rate proportion is a decreasing function of $X_t$ in \eqref{X reflect}, and it has the following asymptotic expansion for small transaction cost $\lambda$. For any fixed time $t\geq 0$,
%$$ \tfrac{\hat{c}_t}{  \hat{\varphi}^{(0)}_t+\hat{\varphi}^{(1)}_t S_t^{(1)}+\hat{\varphi}^{(2)}_t S_t^{(2)}}>c^M.$$
\begin{equation}
\begin{split}\label{asymptotic_c}
\tfrac{\hat{c}_t}{  \hat{\varphi}^{(0)}_t+\hat{\varphi}^{(1)}_t S_t^{(1)}+\hat{\varphi}^{(2)}_t S_t^{(2)}}=c^M + \tfrac{q(1-\rho^2)\sigma_1^2 \zeta^2}{2(1+q)}  \lambda^{\frac{2}{3}} +  O(\lambda), \quad \textrm{a.s.} 
\end{split}
\end{equation}
%(2) The consumption rate proportion is bigger when $X_t=\ux$ than $X_t=\ox$, i.e.,
%$$ \tfrac{\hat{c}_t}{  \hat{\varphi}^{(0)}_t+\hat{\varphi}^{(1)}_t S_t^{(1)}+\hat{\varphi}^{(2)}_t S_t^{(2)}}\Big|_{X_t=\ux}> \tfrac{\hat{c}_t}{  \hat{\varphi}^{(0)}_t+\hat{\varphi}^{(1)}_t S_t^{(1)}+\hat{\varphi}^{(2)}_t S_t^{(2)}}\Big|_{X_t=\ox}.$$
\end{corollary}
\begin{proof}
Using \eqref{strategy} and \eqref{simplie_pi1}, we obtain
  \begin{equation}
\begin{split}\label{consumption_function}
\textrm{(consumption rate proportion at time $t$)}
&=\tfrac{\hat{c}_t}{  \hat{\varphi}^{(0)}_t+\hat{\varphi}^{(1)}_t S_t^{(1)}+\hat{\varphi}^{(2)}_t S_t^{(2)}} \\
&=\tfrac{\hW_t}{g(X_t) \big( e^{-f(X_t)}\pi_1(X_t)\hW_t + (1- \pi_1(X_t))\hW_t  \big)}\\
&=\tfrac{1}{g(x)+\frac{x}{q}(e^{-f(x)}-1)}\Big|_{x=X_t}.
\end{split}
\end{equation}
Under \assref{ass5}, \proref{FBO} and the construction of $f$ in \eqref{fr} imply that $g(x)$ and $e^{-f(x)}$ are increasing functions of $x$ on $[\ux,\ox]$. Then we conclude that the map $x\mapsto \tfrac{1}{g(x)+\frac{x}{q}(e^{-f(x)}-1)}$ is decreasing on $x\in [\ux,\ox]$, due to $f(\ux)=0$, $q>0$, and $\ux>0$. Therefore, \eqref{consumption_function} shows that the optimal consumption rate proportion is a decreasing function of $X_t$, and we obtain the bounds
\begin{align}\label{consumption_bounds}
\tfrac{1}{g(\overline{x})+ \frac{\lambda \overline{x}}{(1-\lambda)q} } \leq  \tfrac{\hat{c}_t}{  \hat{\varphi}^{(0)}_t+\hat{\varphi}^{(1)}_t S_t^{(1)}+\hat{\varphi}^{(2)}_t S_t^{(2)}} \leq  \tfrac{1}{g(\underline{x})}.
\end{align}
Using the asymptotic expansions in \proref{asymptotic_seed}, we derive
\begin{equation}
\begin{split}\label{consumption_asymp}
\tfrac{1}{g(\overline{x})+ \frac{\lambda \overline{x}}{(1-\lambda)q} } &=c^M + \tfrac{q(1-\rho^2)\sigma_1^2 \zeta^2}{2(1+q)}  \lambda^{\frac{2}{3}} +  O(\lambda), \\
\tfrac{1}{g(\underline{x})} &=c^M + \tfrac{q(1-\rho^2)\sigma_1^2 \zeta^2}{2(1+q)}  \lambda^{\frac{2}{3}} +  O(\lambda).
\end{split}
\end{equation}
We conclude \eqref{asymptotic_c} by \eqref{consumption_bounds} and \eqref{consumption_asymp}.
\end{proof}

\corref{c_cor} has the following implications regarding the optimal consumption rate. 
\begin{itemize}
\item[(1)] The proof of \corref{no trading verify} implies that the proportion of wealth invested in the illiquid asset is an increasing function of $X_t$. On the other hand, \corref{c_cor} says that the consumption rate proportion is an decreasing function of $X_t$. Therefore, 
%when the illiquid asset proportion $ \tfrac{ \hat{\varphi}^{(1)}_t S_t^{(1)} }{ \hat{\varphi}^{(0)}_t+\hat{\varphi}^{(1)}_t S_t^{(1)}+\hat{\varphi}^{(2)}_t S_t^{(2)}} $ is close to $\underline{\pi}_1$, the consumption rate proportion $\tfrac{\hat{c}_t}{  \hat{\varphi}^{(0)}_t+\hat{\varphi}^{(1)}_t S_t^{(1)}+\hat{\varphi}^{(2)}_t S_t^{(2)}}$ is larger than the consumption in the situation that the illiquid asset proportion is close to $\overline{\pi}_1$. 
\begin{displaymath}
\begin{split}
&\tfrac{\hat{c}_t}{  \hat{\varphi}^{(0)}_t+\hat{\varphi}^{(1)}_t S_t^{(1)}+\hat{\varphi}^{(2)}_t S_t^{(2)}}\Big|_{\textrm{when  }\tfrac{ \hat{\varphi}^{(1)}_t S_t^{(1)} }{ \hat{\varphi}^{(0)}_t+\hat{\varphi}^{(1)}_t S_t^{(1)}+\hat{\varphi}^{(2)}_t S_t^{(2)}}  \approx \underline{\pi}_1}\\
&> \tfrac{\hat{c}_t}{  \hat{\varphi}^{(0)}_t+\hat{\varphi}^{(1)}_t S_t^{(1)}+\hat{\varphi}^{(2)}_t S_t^{(2)}}\Big|_{\textrm{when  }\tfrac{ \hat{\varphi}^{(1)}_t S_t^{(1)} }{ \hat{\varphi}^{(0)}_t+\hat{\varphi}^{(1)}_t S_t^{(1)}+\hat{\varphi}^{(2)}_t S_t^{(2)}}  \approx \overline{\pi}_1}.
\end{split}
\end{displaymath}
A simple explanation for this phenomenon is that $ \tfrac{ \hat{\varphi}^{(1)}_t S_t^{(1)} }{ \hat{\varphi}^{(0)}_t+\hat{\varphi}^{(1)}_t S_t^{(1)}+\hat{\varphi}^{(2)}_t S_t^{(2)}}  \approx \underline{\pi}_1$ (resp., $\overline{\pi}_1$) implies undrexposure (resp., overexposure) to the risk factor of the illiquid asset, and the investor can increase (resp., decrease) the illiquid asset proportion by increasing (resp., decreasing) consumption. 

\item[(2)] Recall that $c^M$ is the optimal consumption rate proportion if there is no transaction costs. For small enough $\lambda>0$, \corref{c_cor} implies that the consumption rate proportion $\tfrac{\hat{c}_t}{  \hat{\varphi}^{(0)}_t+\hat{\varphi}^{(1)}_t S_t^{(1)}+\hat{\varphi}^{(2)}_t S_t^{(2)}}$ is always bigger than $c^M$. One possible explanation for this effect is that the transaction costs make the investment less attractive, and consequently, induce an increase of the intermediate consumption rate.
\item[(3)] The coefficient of $\lambda^{\frac{2}{3}}$ term in \eqref{asymptotic_c}, $\tfrac{q(1-\rho^2)\sigma_1^2 \zeta^2}{2(1+q)}$, represents the sensitivity of the increase in the consumption rate proportion due to the transaction costs. In other words, the size of $\tfrac{q(1-\rho^2)\sigma_1^2 \zeta^2}{2(1+q)}$ describes how pronounced the effect described in (2) is. To discuss the contributions of the liquid risky asset trading opportunity on this effect, we consider the model with the illiquid asset only (without the liquid risky asset) and compare the optimal consumption in that model with the consumption in our model (with the liquid risky asset). The case (2) in \remref{special cases} and the corresponding discussion in the Appendix imply that when $\mu_2=\rho=0$, our market model is equivalent to the market model with the illiquid asset only. We focus on two special cases, $\mu_2=\frac{\rho\sigma_2 \mu_1}{\sigma_1}$ and $\rho=0$, and compare the $\lambda^{\frac{2}{3}}$-coefficients for the market model {\it with} or {\it without} the liquid risky asset\footnote{
The inequalities \eqref{c_compare2} and \eqref{c_compare1} are followed by the following observation.
$$
\big(\tfrac{q(1-\rho^2)\sigma_1^2 \zeta^2}{2(1+q)}\big)^3- \big( \tfrac{q(1-\rho^2)\sigma_1^2 \zeta^2}{2(1+q)}\big|_{\mu_2=\rho=0}\big)^3=
\begin{cases}
-\tfrac{9 \mu_1^2 \mu_2^2 q^3 (1+q)^3(\mu_1(1+q)-\sigma_1^2)^4}{128\sigma_1^8\sigma_2^2}\leq 0, &\textrm{if  $\mu_2=\tfrac{\rho\sigma_2 \mu_1}{\sigma_1}$},\\
\frac{9\mu_1^4 \mu_2^2 q^3 (1+q)^5 (\mu_2^2(1+q)^2 \sigma_1^2 + 2(\mu_1(1+q)-\sigma_1^2)^2 \sigma_2^2)}{128 \sigma_1^8 \sigma_2^4}
 \geq 0, &\textrm{if  $\rho=0$}.
 \end{cases}
$$
} 
\begin{align}
&\tfrac{q(1-\rho^2)\sigma_1^2 \zeta^2}{2(1+q)} \leq \Big( \tfrac{q(1-\rho^2)\sigma_1^2 \zeta^2}{2(1+q)}\Big|_{\mu_2=\rho=0}\Big), \quad \textrm{if  $\mu_2=\tfrac{\rho\sigma_2 \mu_1}{\sigma_1}$}.\label{c_compare2}\\
&\tfrac{q(1-\rho^2)\sigma_1^2 \zeta^2}{2(1+q)} \geq \Big(\tfrac{q(1-\rho^2)\sigma_1^2 \zeta^2}{2(1+q)}\Big|_{\mu_2=\rho=0} \Big), \quad \textrm{if  $\rho=0$.}\label{c_compare1}
\end{align}
\begin{itemize}
\item[$\bullet$] In case $\mu_2=\tfrac{\rho\sigma_2 \mu_1}{\sigma_1}$: The inequality \eqref{c_compare2} implies that the effect of transaction costs on the consumption is {\it less} pronounced if the investor can trade the liquid risky asset. We observe that $\pi_2^M=0$ in this case, i.e., it is optimal not to invest in the liquid risky asset at all if there is no transaction costs. 
%The existence of transaction costs increases the consumption rate proportion in both markets (with or without the liquid risky asset) as described in (2). 
In the market with the liquid risky asset, the investor can adjust the exposure to the risk factor of the illiquid asset by trading the (correlated) liquid risky asset. Consequently, the liquid risky asset trading opportunity induces less frequent trading of the illiquid asset, and mitigates the effect of the transaction costs on the consumption. %when the optimal investment in the liquid risky asset is small enough ($\pi_2^M\approx 0$), 
 \item[$\bullet$] In case $\rho=0$: The inequality \eqref{c_compare1} implies that the effect of the transaction costs on the consumption is $more$ pronounced if the investor can trade the liquid risky asset. Because $\rho=0$, $\pi_1^M$ is the optimal investment proportion of the illiquid asset for both markets (with or without the liquid risky asset), i.e., $\pi_1^M=\pi^M$. The existence of the liquid risky asset induces the investor to also take an exposure to the risk factor of the liquid risky asset ($\pi_2^M \neq 0$), thereby increasing the volatility of the total wealth process. Therefore, the model with liquid asset trading opportunity has more frequent trading of the illiquid asset (see \corref{no trading verify}) and more trading costs due to rebalancing. Consequently, the effect of the transaction costs on consumption is stronger in the market with the liquid asset.
\end{itemize}
For general parameters, the liquid asset trading opportunity has both of the above aspects (adjustment of the illiquid asset risk factor, and increase in the total wealth volatility). These aspects compete and determine the size of the effect of the transaction costs on consumption when the liquid risky asset becomes tradable in the market.
\end{itemize}

\section{proofs}

This section is devoted to the proof of \proref{FBO} and \thmref{well-posed thm}. We split the proof of \proref{FBO} into five propositions which take care of different parameter regimes as following.
\begin{itemize}
\item \proref{hard}: $0<p<1$, $\mu_1>\frac{\rho \mu_2 \sigma_1}{\sigma_2}$ and $\delta > \tfrac{q}{2(1-\rho^2)} \big( (\tfrac{\mu_1}{\sigma_1})^2+ (\tfrac{\mu_2}{\sigma_2})^2- 2\rho \tfrac{\mu_1 \mu_2}{\sigma_1 \sigma_2} \big)$.
\item \proref{hard2}: $0<p<1$, $\mu_1>\frac{\rho \mu_2 \sigma_1}{\sigma_2}$ and $\delta  \leq \tfrac{q}{2(1-\rho^2)} \big( (\tfrac{\mu_1}{\sigma_1})^2+ (\tfrac{\mu_2}{\sigma_2})^2- 2\rho \tfrac{\mu_1 \mu_2}{\sigma_1 \sigma_2} \big)$.
\item \proref{pro2}: $0<p<1$ and $\mu_1<\frac{\rho \mu_2 \sigma_1}{\sigma_2}$.
\item \proref{p nega}:  $p<0$ and $\mu_1>\frac{\rho \mu_2 \sigma_1}{\sigma_2}$.
\item \proref{pro4}: $p<0$ and $\mu_1<\frac{\rho \mu_2 \sigma_1}{\sigma_2}$. 
\end{itemize}
We provide detailed proof for \proref{hard} and \proref{hard2}. Proofs for the other cases are similar, hence we omit some of the details. Finally, the proof of \thmref{well-posed thm} is given at the end of this section.

\smallskip

We use the following notation for convenience. Using the expression \eqref{optimizer} of optimizers, we rewrite \eqref{inf_ode} as
\begin{equation}
\begin{split}\label{ode}
\tfrac{A(x,g(x))g'(x)^2 + B(x,g(x)) g'(x) + C(x,g(x))}{2(1+q) \sigma_2^2 (1+g'(x))\big(q g(x)(1+g'(x)) - (1+q) x g'(x)\big)}=0,
\end{split}
\end{equation}
where
{
\begin{equation}
\begin{split}\label{ABC}
A(x,y)&=-2(1+q)^2\sigma_2^2 \sgn(p) x\\
& + \big((1+q)^4 \mu_2^2 - 2\rho q(1+q)^2 \sigma_1\sigma_2 \mu_2+ ((1+2q+q^2 \rho^2)\sigma_1^2 -2(1+q)^2 \mu_1) \sigma_2^2     \big)  x^2\\
&+ (1+q)\Big( 2q\sigma_2^2 \sgn(p) + \big(2\rho q^2  \sigma_1 \sigma_2 \mu_2 - 2q (1+q)^2 \mu_2^2+ (2\delta(1+q)^2 +q(2\mu_1 - \sigma_1^2))\sigma_2^2   \big)x        \Big)y \\ 
& - q(1+q)^2 (2\delta \sigma_2^2 - q \mu_2^2) y^2\\
B(x,y)&=-2(1+q)^2\sigma_2^2 \sgn(p) x + 2\sigma_2\big( \rho \sigma_1\mu_2 (1+q)^2 - (\mu_1 (1+q)^2 -q(1-\rho^2)\sigma_1^2)\sigma_2     \big)x^2 \\
&+(1+q)\Big( 4q \sigma_2^2 \sgn(p) +\big(2\rho q(q-1)\sigma_1\sigma_2 \mu_2 -2q(1+q)^2\mu_2^2 + (2\delta(1+q)^2 + q(4\mu_1-\sigma_1^2))\sigma_2^2     \big) x     \Big) y\\
&-2q(1+q)^2 (2\delta \sigma_2^2-q\mu_2^2) y^2\\
C(x,y)&=-(1-\rho^2)\sigma_1^2\sigma_2^2 x^2 + 2q(1+q)\sigma_2 \big( \sgn(p) \sigma_2 + (\mu_1 \sigma_2- \rho \mu_2 \sigma_1)x \big)y \\
&- q(1+q)^2 (2\delta \sigma_2^2 - q \mu_2^2) y^2\\
\end{split}
\end{equation}
}

Also, let $\Delta_x$ be the discriminant of quadratic equation with respect to $x$, i.e., 
$$\Delta_x ( a x^2+bx+c):=b^2-4ac.$$

Constants $y_C$, $x_D$, $y_D$, $x_M$ and $y_M$ are defined as
\begin{equation}
\begin{split}
y_C&:=\tfrac{2 \sigma_2^2\sgn(p) }{(1+q)(2\delta \sigma_2^2 - q \mu_2^2)},\\
x_D&:=\tfrac{2q(1+q)\sgn(p)}{2\delta(1+q)^2 + q(\sigma_1^2-2(1+q)\mu_1)}, \\
y_D&:=\tfrac{2(1+q)\sgn(p)}{2\delta(1+q)^2 + q(\sigma_1^2-2(1+q)\mu_1)},\\
x_M&:=\tfrac{q (\mu_1-\frac{\rho\sigma_1\mu_2}{\sigma_2})\sgn(p)}{(1-\rho^2)\sigma_1^2(\delta -\frac{q}{2(1-\rho^2)} ( (\frac{\mu_1}{\sigma_1})^2+ (\frac{\mu_2}{\sigma_2})^2- 2\rho \frac{\mu_1 \mu_2}{\sigma_1 \sigma_2} ) )},\\
y_M&:=\tfrac{\sgn(p)}{(1+q)\big( \delta -\frac{q}{2(1-\rho^2)} ( (\frac{\mu_1}{\sigma_1})^2+ (\frac{\mu_2}{\sigma_2})^2- 2\rho \frac{\mu_1 \mu_2}{\sigma_1 \sigma_2} )\big) }.
\end{split}
\end{equation}

\begin{lemma}\label{common} Let $A,B,C$ be functions as in \eqref{ABC}. Then the following statements hold:\\
(1) $\{(x,y):\,  B(x,y)=C(x,y)=0 \}=\{(x,y):\,  B(x,y)=A(x,y)=0 \}=\{(0,0),(0,y_C)\}$.\\
(2) $\{(x,y):\, x \neq 0, \,\, B(x,y)^2-4A(x,y)C(x,y) =0\} = \{(x_D,y_D)\}$.\\
(3) $\{(x,y):\, x \neq 0, \,\, B(x,y)^2-4A(x,y)C(x,y) <0\} = \emptyset$.
\end{lemma}
\begin{proof}
(1) Suppose that $(x,y)\in \bR^2$ satisfies $B(x,y)=C(x,y)=0$. Then 
\begin{displaymath}
\begin{split}
0&=2C(x,y)-B(x,y)\\
&=(1+q)x \Big( \big(2q(1+q)^2\mu_2^2 - 2\rho q(1+q)\sigma_1 \sigma_2 \mu_2+(q \sigma_1^2-2\delta (1+q)^2)\sigma_2^2\big) y \\
& \quad + 2\sigma_2 \big((1+q)\sigma_2 \sgn(p) + \big( \sigma_2((1+q)\mu_1 - (1-\rho^2)\sigma_1^2)-\rho(1+q) \sigma_1 \mu_2  \big)x         \big)           \Big)
\end{split}
\end{displaymath}
There are three possibilities:\\
(i) In case $2q(1+q)^2\mu_2^2 - 2\rho q(1+q)\sigma_1 \sigma_2 \mu_2+(q \sigma_1^2-2\delta (1+q)^2)\sigma_2^2=0$ and $(1+q)\sigma_2 \sgn(p) + \big( \sigma_2((1+q)\mu_1 - (1-\rho^2)\sigma_1^2)-\rho(1+q) \sigma_1 \mu_2  \big)x =0$, we solve these equations for $\delta$ and $x$ and obtain
$$\delta=\tfrac{q(2(1+q)^2\mu_2^2 - 2\rho (1+q)\sigma_1 \sigma_2 \mu_2 + \sigma_1^2 \sigma_2^2)}{2(1+q)^2 \sigma_2^2} \textrm{  and  }   x=\tfrac{(1+q)\sigma_2 \sgn(p)}{\rho (1+q)\sigma_1 \mu_2 - (\mu_1(1+q)-(1-\rho^2)\sigma_1^2)\sigma_2}.$$ Substituting these expressions for $\delta$ and $x$, we obtain
$$\Delta_y (C(x,y))=-(1-\rho^2)\Big(\tfrac{2q(1+q)\sigma_1 \sigma_2^2((1+q)\mu_2-\rho \sigma_1 \sigma_2)}{\rho (1+q)\sigma_1 \mu_2 - (\mu_1(1+q)-(1-\rho^2)\sigma_1^2)\sigma_2}\Big)^2 <0.$$
Therefore, $C(x,y)\neq0$, which is a contradiction.\\
(ii) In case $y=-\frac{2\sigma_2 \big((1+q)\sigma_2 \sgn(p) + \big( \sigma_2((1+q)\mu_1 - (1-\rho^2)\sigma_1^2)-\rho(1+q) \sigma_1 \mu_2  \big)x         \big) }{2q(1+q)^2\mu_2^2 - 2\rho q(1+q)\sigma_1 \sigma_2 \mu_2+(q \sigma_1^2-2\delta (1+q)^2)\sigma_2^2}$, we substitute this expression for $y$ and obtain
$$\Delta_x(C(x,y))=-(1-\rho^2)\Big(\tfrac{4q(1+q)\sigma_1 \sigma_2^3 ((1+q)\mu_2 - \rho \sigma_1 \sigma_2)  }{2q(1+q)^2 \mu_2^2 - 2\rho q(1+q)\sigma_1 \sigma_2 \mu_2+(q \sigma_1^2 - 2\delta (1+q)^2)\sigma_2^2}\Big)^2<0.$$
Therefore, $C(x,y)\neq0$, which is a contradiction.\\
(iii) In case $x=0$, we solve $B(0,y)=C(0,y)=0$ for $y$ and obtain $y=0$ or $y=y_C$.
Therefore, $\{(x,y):\,  B(x,y)=C(x,y)=0 \}=\{(0,0),(0,y_C)\}$. 

The proof of $\{(x,y):\, B(x,y)=A(x,y)=0 \}=\{(0,0),(0,y_C)\}$ is similar.

\bigskip

(2) $B(x,y)^2-4A(x,y)C(x,y)$ is quadratic in $y$. If $x\neq 0$ and $x\neq x_D$, then
\begin{displaymath}
\begin{split}
&\Delta_y(B(x,y)^2-4A(x,y)C(x,y))\\
&=-(1-\rho^2)\Big( 4(1+q)^2 \sigma_1 \sigma_2^3 ((1+q)\mu_2 - \rho \sigma_1 \sigma_2)x^2 \big( 2\delta (1+q)^2 +q(\sigma_1^2-2\mu_1(1+q))\big)(x-x_D) \Big)^2 < 0.
\end{split}
\end{displaymath}
Hence $B(x,y)^2-4A(x,y)C(x,y)$ can be zero only when $x=0$ or $x=x_D$. We have
 \begin{displaymath}
\begin{split}
&B(x_D,y)^2-4A(x_D,y)C(x_D,y)\\
&=\sigma_2^2 x_D^2 \Big( (1-\rho^2)\sigma_2^2 \big(2\delta(1+q)^2-q\sigma_1^2 \big)^2 + \big(2q(1+q)\sigma_1 \mu_2 - \rho \sigma_2(2\delta(1+q)^2+q \sigma_1^2)   \big)^2  \Big) (y-y_D)^2,
\end{split}
\end{displaymath}
and observe that $2\delta(1+q)^2-q\sigma_1^2$ and $2q(1+q)\sigma_1 \mu_2 - \rho \sigma_2(2\delta(1+q)^2+q \sigma_1^2)$ cannot be zero at the same time because $\mu_2 \neq \frac{\rho \sigma_1 \sigma_2}{1+q}$.
Now we conclude that $$\big\{(x,y):\, x \neq 0, \,\, B(x,y)^2-4A(x,y)C(x,y) =0\big\} = \big\{(x_D,y_D)\big\}$$ 

\bigskip

(3) In the proof of (2), we can see that \\
(i) if $x=x_D$, then $B(x,y)^2-4A(x,y)C(x,y)\geq 0$, \\
(ii) if $x\neq 0$ or $x\neq x_D$, then there is no solution to $B(x,y)^2-4A(x,y)C(x,y)=0$. \\
Therefore, to prove (3), it is enough to show that the coefficient of $y^2$ in $B(x,y)^2-4A(x,y)C(x,y)$ is positive for any $x \neq 0$. Indeed, the coefficient of $y^2$ in $B(x,y)^2-4A(x,y)C(x,y)$ can be written as a sum of squares
 \begin{displaymath}
\begin{split}
((1+q)\sigma_2 x)^2 \Big( (1-\rho^2)\sigma_2^2 \big(2\delta(1+q)^2-q\sigma_1^2 \big)^2 + \big(2q(1+q)\sigma_1 \mu_2 - \rho \sigma_2(2\delta(1+q)^2+q \sigma_1^2)   \big)^2  \Big)>0.
\end{split}
\end{displaymath}
\end{proof}

\begin{proposition}\label{hard} If $0<p<1$, $\mu_1>\frac{\rho \mu_2 \sigma_1}{\sigma_2}$ and $\delta > \tfrac{q}{2(1-\rho^2)} \big( (\tfrac{\mu_1}{\sigma_1})^2+ (\tfrac{\mu_2}{\sigma_2})^2- 2\rho \tfrac{\mu_1 \mu_2}{\sigma_1 \sigma_2} \big)$, then \proref{FBO} holds.
\end{proposition}
\begin{proof}
We observe that $A,B,C$ are quadratic in $x,y$. The level curve $C=0$ is an ellipse, because
$$\tfrac{\partial^2 C}{\partial x \partial y}-4\, \tfrac{\partial^2 C}{\partial x^2}\tfrac{\partial^2 C}{\partial y^2}= -8(1-\rho^2)q(1+q)^2 \sigma_1^2 \sigma_2^4 \Big(\delta -\tfrac{q}{2(1-\rho^2)} ( (\tfrac{\mu_1}{\sigma_1})^2+ (\tfrac{\mu_2}{\sigma_2})^2- 2\rho \tfrac{\mu_1 \mu_2}{\sigma_1 \sigma_2} ) \Big)<0.$$ 
\assref{ass} implies that $y_C,x_D, y_D,x_M, y_M$ are all positive. The slopes of the level curves $C=0$, $B=0$ and $A=0$ at the points $(0,y_C)$ and $(0,0)$ are 
\begin{displaymath}
\begin{split}
\begin{cases}C=0: &\frac{2\sigma_2^2 (\mu_1 -\frac{\rho \sigma_1 \mu_2}{\sigma_2})}{(1+q)(2\delta \sigma_2^2- q \mu_2^2)} \textrm{  and  }0 \\
B=0: &\frac{-\mu_2^2 (1+q)^2 -2(1-q)\rho \sigma_1 \sigma_2 \mu_2 +(4\mu_1- \sigma_1^2)\sigma_2^2    }{2(1+q)(2\delta \sigma_2^2- q \mu_2^2)} \textrm{  and  }\frac{1+q}{2q}\\
A=0: &\frac{-\mu_2^2 (1+q)^2 + 2\rho q \sigma_1 \sigma_2 \mu_2 + (2\mu_1 - \sigma_1^2)\sigma_2^2}{(1+q)(2\delta \sigma_2^2- q \mu_2^2)} \textrm{  and  }\frac{1+q}{q}\\
\end{cases}
\end{split}
\end{displaymath}
We observe that 
\begin{displaymath}
\begin{split}
&\tfrac{2\sigma_2^2 (\mu_1 -\frac{\rho \sigma_1 \mu_2}{\sigma_2})}{(1+q)(2\delta \sigma_2^2- q \mu_2^2)}- \tfrac{-\mu_2^2 (1+q)^2 -2(1-q)\rho \sigma_1 \sigma_2 \mu_2 +(4\mu_1- \sigma_1^2)\sigma_2^2    }{2(1+q)(2\delta \sigma_2^2- q \mu_2^2)} = \tfrac{\mu_2^2 (1+q)^2 - 2\mu_2 (1+q) \rho \sigma_1 \sigma_2 + \sigma_1^2 \sigma_2^2}{2(1+q)(2\delta \sigma_2^2- q \mu_2^2) }>0,\\
& \tfrac{-\mu_2^2 (1+q)^2 -2(1-q)\rho \sigma_1 \sigma_2 \mu_2 +(4\mu_1- \sigma_1^2)\sigma_2^2    }{2(1+q)(2\delta \sigma_2^2- q \mu_2^2)}-\tfrac{-\mu_2^2 (1+q)^2 + 2\rho q \sigma_1 \sigma_2 \mu_2 + (2\mu_1 - \sigma_1^2)\sigma_2^2}{(1+q)(2\delta \sigma_2^2- q \mu_2^2)}= \tfrac{\mu_2^2 (1+q)^2 - 2\mu_2 (1+q) \rho \sigma_1 \sigma_2 + \sigma_1^2 \sigma_2^2}{2(1+q)(2\delta \sigma_2^2- q \mu_2^2) }>0,
\end{split}
\end{displaymath}
where we use the Cauchy-Schwarz inequality. By this observation and \lemref{common} (1), the quadratic curves $A=0, B=0,C=0$ are as in Figure 1. Using the observation that the coefficients of $y^2$ in $A,B,C$ are all negative, we partition the ellipse as 
\begin{displaymath}
\begin{split}
&\{(x,y):\, x>0,\, C(x,y)\geq0 \}=\Omega_1 \cup \Omega_2 \cup \Omega_3\cup \Omega_4,\\
&\textrm{where  }
\begin{cases}
\Omega_1:=\{(x,y):\, x>0, \,\,  C>0,\, B\geq 0,\, A\geq 0\} \\
\Omega_2:=\{(x,y):\, x>0, \,\,  C>0,\, B\geq 0,\, A< 0\} \\
\Omega_3:=\{(x,y):\, x>0, \,\,  C\geq 0,\, B< 0,\, A< 0\} \\
\Omega_4:=\{(x,y):\, x>0, \,\,  C\geq 0,\, B< 0,\, A\geq 0\} \\
\end{cases}
\end{split}
\end{displaymath}

\begin{figure}
\begin{center}
  \includegraphics[width=6cm]{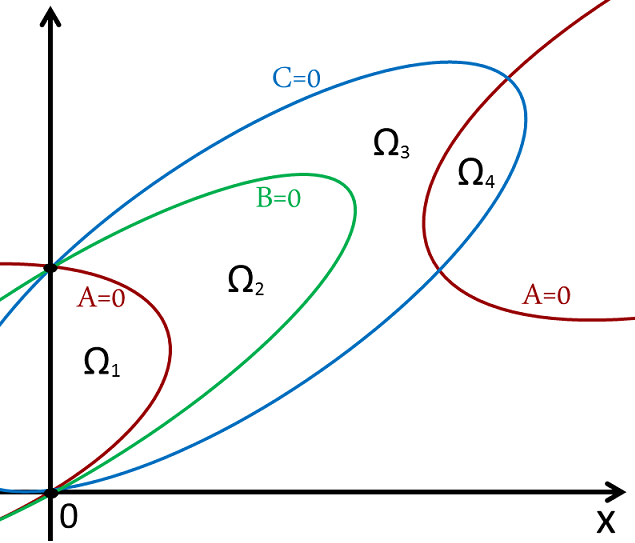}\qquad \includegraphics[width=6cm]{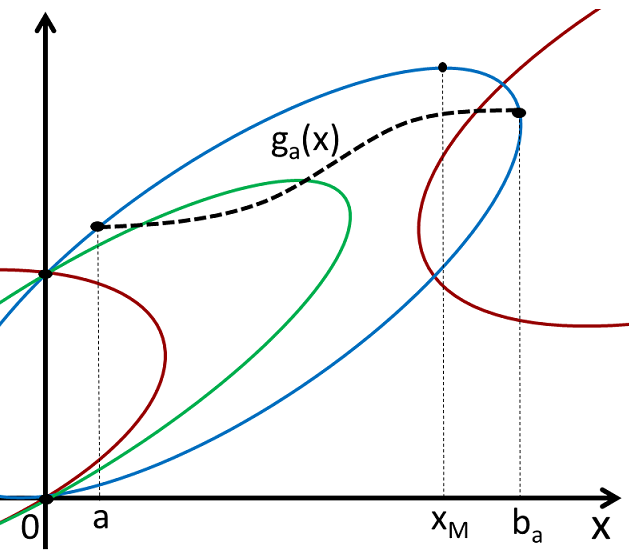}
\end{center}
\figcaption{$0<p<1$, $\mu_1>\frac{\rho \mu_2 \sigma_1}{\sigma_2}$ and $\delta > \tfrac{q}{2(1-\rho^2)} \big( (\tfrac{\mu_1}{\sigma_1})^2+ (\tfrac{\mu_2}{\sigma_2})^2- 2\rho \tfrac{\mu_1 \mu_2}{\sigma_1 \sigma_2} \big)$}
\end{figure}

By \lemref{common} (3), $\sqrt{B(x,y)^2-4 A(x,y) C(x,y)}$ is well defined if $x\neq0$. If $(x,y)\in \Omega_2\cup \Omega_3 \cup \Omega_4$, then $-B(x,y)+\sqrt{B(x,y)^2-4 A(x,y) C(x,y)}>0$. 

\bigskip

{\bf Claim 1}: \emph{For any $a\in \bR$ such that $0<a<x_M$, there exists a constant $b_a>a$ and a function $g_a:[a,b_a]\mapsto \bR$ such that
\begin{equation}
\begin{split}\label{FG}
&g_a'(x)=F(x,g_a(x)), \quad g_a(a)=\Gamma(a), \quad g_a'(b_a)=0,  \\
\textrm{where } &\begin{cases} F(x,y):=\tfrac{2C(x,y)}{-B(x,y)+\sqrt{B(x,y)^2-4 A(x,y) C(x,y)}},\\
\Gamma(x):=\tfrac{q\sigma_2 (\sigma_2+(\mu_1 \sigma_2 - \rho \sigma_1 \mu_2)x) + \sigma_2\sqrt{ 2q\delta(\rho^2-1)\sigma_1^2\sigma_2^2 x^2 +  q^2( \mu_2^2 \sigma_1^2 x^2 - 2\rho \sigma_1 \sigma_2 \mu_2 x(1+\mu_1 x) + \sigma_2^2(1+\mu_1 x)^2)  }     }{q(1+q)(2 \delta \sigma_2^2 - q \mu_2^2)}.
\end{cases}
\end{split}
\end{equation}
Here, $y=\Gamma(x)$ is the equation for the upper part of the ellipse $C=0$.}

(Proof of Claim 1): The function $F$ is continuous and nonnegative on $\Omega_2\cup \Omega_3 \cup \Omega_4$. By the Peano existence theorem, starting from $(a,\Gamma(a))$, we can evolve the above ODE to the right (see Figure 1) until $(x,g_a(x))$ reaches the boundary of $\Omega_2\cup \Omega_3 \cup \Omega_4$. Indeed, the curve $(x,g_a(x))$ is inside of $\Omega_2\cup \Omega_3 \cup \Omega_4$ for $x>a$ close enough to $a$, because $\Gamma'(a)>0$ and $g_a'(a)=F(a,\Gamma(a))=0$. Let the equation of the upper curve of $\partial \Omega_1 \cap \partial \Omega_2$ be $l(x)$. Observe that $(a,g_a(a))$ is above the curve $\partial \Omega_1 \cap \partial \Omega_2$, i.e., $g_a(a)>l(a)$ (see Figure 1). Define $b_a>a$ as
$$b_a:=\inf \big\{x>a : \, (x,g_a(x))\in \partial \big(\Omega_2\cup \Omega_3 \cup \Omega_4 \big) \big\}.$$
As $F\geq 0$ on $\Omega_2\cup \Omega_3 \cup \Omega_4$, $\lim_{x\uparrow b_a}g_a(x)$ exists. Suppose that $\lim_{x\uparrow b_a}g_a(x)=l(b_a)$. The definition of $b_a$ implies that $g_a'(x)>0$ for $a<x<b_a$, and we observe that $-B+\sqrt{B^2-4AC}=0$ on $\partial \Omega_1 \cap \partial \Omega_2$. Then, we produce
$$+\infty=\lim_{x \uparrow b_a}F(x,g_a(x))=\lim_{x \uparrow b_a} g_a'(x)
=\lim_{x \uparrow b_a} \tfrac{g_a(b_a)-g_a(x)}{b_a-x} \leq \lim_{x\uparrow b_a} \tfrac{l(b_a)-l(x)}{b_a-x}=l'(b_a)\leq l'(0)<\infty,
$$
where we use L'Hopital's Rule and concavity of the curve $l(x)$. This is a contradiction, and we conclude that 
$$b_a=\inf \{x>a : \, C(x,g_a(x))=0 \}.$$
Then $g_a'(b_a)=F(b_a,g_a(b_a))=0$ because $C(b_a,g_a(b_a))=0$. \\
(End of the proof of Claim 1).

\bigskip

{\bf Claim 2}: \emph{$g_a'(x)\neq \frac{1}{q}$  for  $x\in[a,b_a]$.}

(Proof of Claim 2): Suppose that $\{x\in [a,b_a]:\, g_a'(x)=\frac{1}{q}\} \neq \emptyset$. Then we define $x_0$ as 
$$x_0:=\inf \{x\in [a,b_a]:\, g_a'(x)=\tfrac{1}{q}\}. $$
$F(x_0,y)=\frac{1}{q}$ implies $\big(2q C(x_0,y)+B(x_0,y)\big)^2=B(x_0,y)^2-4 A(x_0,y) C(x_0,y).$
Solving this equation for $y$, we obtain
\begin{equation}
\begin{split}\label{lines}
y=\tfrac{x_0}{q}\quad \textrm{or} \quad y=L(x_0),  \textrm{  where  } L(x):=\tfrac{2(1+q)\sigma_2^2-(\mu_2^2(1+q)^2 +(\sigma_1^2 -2(1+q)\mu_1)\sigma_2^2 )x}{(1+q)^2 (2\delta \sigma_2^2 -q \mu_2^2)}.
\end{split}
\end{equation}
Direct computation produces 
\begin{displaymath}
\begin{split}
x_0>0 \textrm{  and  }F(x_0,\tfrac{x_0}{q})=\tfrac{1}{q} \quad \iff \quad 0< x_0 \leq x_D\\
x_0>0 \textrm{  and  }F(x_0,L(x_0))=\tfrac{1}{q} \quad \iff \quad 0< x_0 \leq x_D\\
\end{split}
\end{displaymath}
Therefore, $x_0\leq x_D$. And the point $(x_0,g_a(x_0))$ is on one of the two lines in \eqref{lines}. We also check that
\begin{equation}
\begin{split}\label{inside1}
A(x_D,y_D)>0,\, B(x_D,y_D)<0,\, C(x_D,y_D)>0 \quad \Longrightarrow \quad (x_D,y_D)\in\Omega_4 .
\end{split}
\end{equation}
\eqref{inside1} implies that the line segment $y=L(x)$ connecting $(0,y_0)$ and $(x_D,y_D)$ is inside of the ellipse $C=0$. Therefore, $(a,g_a(a))$ is above this line segment. 
If $(x_0,g_a(x_0))$ is on this line segment, then by the definition of $x_0$, $g_a'(x_0)\leq L'(x_0)$. This is a contradiction because
$$g_a'(x_0)- L'(x_0)=\tfrac{1}{q}-\tfrac{(2(1+q)\mu_1-\sigma_1^2)\sigma_2^2 - (1+q)^2 \mu_2^2}{(1+q)^2(2\delta \sigma_2^2 - q \mu_2^2)}=\tfrac{2\delta(1+q)^2 + q(\sigma_1^2-2(1+q)\mu_1)}{q(1+q)^2(2\delta \sigma_2^2 - q \mu_2^2)}>0.$$
The line segment $y=\tfrac{x}{q}$ connecting $(0,0)$ and $(x_D,y_D)$ is below the line $y=L(x)$. Therefore, $(x_0,g_a(x_0))$ cannot be on this line segment, either. Now we reach the contradiction.\\
(End of the proof of Claim 2).

\bigskip

{\bf Claim 3}: \emph{In Claim 1, the solution $g_a$ is unique in $x\in[a,b_a]$ and $g_a\in C^2([a,b_a])$}

(Proof of Claim 3): Direct computation produces $F(x_D,y_D)=\frac{1}{q}$. Therefore, by Claim 2, $(x,g_a(x)) \neq (x_D,y_D)$ for $x\in [a,b_a]$. Considering \lemref{common} (2) and (3), we can include the set $\{(x,g_a(x)): x\in [a,b_a]\}$ in a compact set in $\bR^2$ where $F$ is uniformly Lipschitz. Then the uniqueness follows from the Picard-Lindelof €"theorem. Because $F\in C^2(\bR_+ \times \bR \setminus \{(x_D,y_D)\} )$, we conclude $g_a\in C^2([a,b_a])$.\\
(End of the proof of Claim 3).

\bigskip

{\bf Claim 4}: \emph{Let $G(a):=\int_a^{b_a} \frac{g_a'(x)}{x} dx $. Then $G$ has the following properties: \\
(i) $G$ is continuous on $(0,x_M)$.\\
(ii) $\lim_{a\uparrow x_M} G(a)=0$.\\
(iii) $\lim_{a\downarrow 0} G(a)=\infty$.
}

(Proof of Claim 4): (i) Suppose that $g_a(\cdot)$ is tangent to the ellipse $C=0$ at $x=b_a$. Because $g_a'(b_a)=0$, the only possibility is $(b_a,g_a(b_a))=(x_M,y_M)$. Direct computation produces $g_a''(b_a)=\frac{d}{dx}F(x,g_a(x))|_{x=x_M}=0$, $\Gamma'(x_M)=0$ and $\Gamma''(x_M)=-\frac{(1-\rho^2)\sigma_1^2}{q(1+q)}<0$. This observation implies that $g_a(x)>\Gamma(x)$ for $x<x_M$ close enough to $x_M$, which is a contradiction. Therefore, $g_a$ is not tangent to $C=0$ at $x=b_a$, and by the implicit function theorem and the continuity of $g_a$ with respect to the initial data (see, e.g., Theorem VI., p 145 in \cite{Wal98}), the map $a\mapsto b_a$ is continuous.

By Claim 2, $0\leq g_a'<\frac{1}{q}$ on $[a,b_a]$ for any $a\in(0,x_M)$. Because the map $a\mapsto b_a$ is continuous, $G$ is continuous on $a\in(0,x_M)$ by the dominated convergence theorem. 

(ii) The ellipse $C=0$ has the biggest $y$ value at $x=x_M$. Because $g_a$ increases and $g_a'(b_a)=0$, we have $b_a \geq x_M$ and $\lim_{a\uparrow x_M}b_a =x_M$. As $|g_a'| \leq \frac{1}{q}$, the dominated convergence theorem produces
$$\lim_{a\uparrow x_M} |G(a)|\leq \lim_{a\uparrow x_M} \int_a^{b_a} \tfrac{|g_a'(x)|}{x} dx \leq \lim_{a\uparrow x_M} \tfrac{1}{q} \ln(\tfrac{b_a}{a})=0.$$

(iii) Let $\Gamma_k:\bR \mapsto \bR$ for $k\geq 0$ be defined as
$$\Gamma_k(x):=\max\{y: F(x,y)=k \}.$$
We observe that $\Gamma_k(x)$ is well-defined for small enough $k$ and $x$, and $\Gamma_k(x)<\Gamma(x)$ for $x>0$ and $k>0$, in the intersection of their domains. Also, $\Gamma_0(x)=\Gamma(x)$ and $\Gamma_0'(0)=\frac{2\sigma_2(\sigma_2\mu_1-\rho \sigma_1 \mu_2)}{(1+q)(2\delta \sigma_2^2 -q \mu_2^2}>0$. Because $\Gamma_k'(x)$ is jointly continuous for small enough $x$ and $k$, there exists $\epsilon>0$ such that
$$\Gamma_\epsilon'(x)>2\epsilon \quad \textrm{for}\quad x\in [0,\epsilon].$$
We define $h(x):=x+ \frac{\Gamma(x)-\Gamma_\epsilon (x)}{\epsilon}$. Because $\Gamma_\epsilon(0)=\Gamma(0)$, we have $\lim_{x\downarrow 0} h(x)=0$. Hence, there exists $a_\epsilon>0$ such that $0<h(x)<\epsilon$ for their $x\leq a_\epsilon$.

Let $a\in (0, a_\epsilon)$ be fixed. Suppose that $g_a(x)>\Gamma_\epsilon(x)$ for $x\in [a,h(a)]$. Then, the definition of the level curve $\Gamma_\epsilon$ implies that $g_a'(x)<\epsilon$ on $[a, h(a)]$, but this is a contradiction
$$ 0<g_a(h(a))-\Gamma_\epsilon (h(a))= \int_a^{h(a)} \big(g_a'(x)-\Gamma_\epsilon '(x)\big)\, dx  + \Gamma(a)-\Gamma_\epsilon(a) \leq -\epsilon (h(a)-a)+ \Gamma(a)-\Gamma_\epsilon(a)=0. $$
Therefore, $g_a$ intersect $\Gamma_\epsilon$ on $[a,h(a)]$. After $g_a$ intersect $\Gamma_\epsilon$, $g_a$ is below $\Gamma_\epsilon$ for a while: $g_a(x)<\Gamma_\epsilon (x)$ for $x\in [h(a),\epsilon]$, because $F(x,\Gamma_\epsilon(x))=\epsilon< \Gamma_\epsilon'(x)$ for $x\in [0,\epsilon]$. This means that $g_a'(x)\geq \epsilon$ for $x\in [h(a),\epsilon]$. 

Now we take the limit $a\downarrow 0$ and obtain the result
$$\liminf_{a\downarrow 0} G(a) \geq \liminf_{a\downarrow 0} \int_{h(a)}^{\epsilon} \frac{\epsilon}{x}\, dx =\liminf_{a\downarrow 0} \epsilon \ln\big(\tfrac{\epsilon}{h(a)}\big)=\infty.$$ 
(End of the proof of Claim 4).\\

By Claim 4 and the intermediate value theorem, we can choose $a\in (0,x_M)$ such that 
$$\int_a^{b_a} \frac{g_a'(x)}{x} \, dx =\ln (\tfrac{1+\overline{\lambda}}{1-\underline{\lambda}}).$$ 
We set $\ux=a$, $\ox=b_a$ and $g=g_a$. Then, $\ux$, $\ox$ and $g$ satisfy
\begin{equation}
\begin{split}\label{g1}
&\textrm{$0<\ux<\ox$, $g\in C^2([\ux,\ox])$ and}\\
&g'(x)=F(x,g(x)), \quad g'(\ux)=0,\quad g'(\ox)=0, \quad \int_{\ux}^{\ox} \frac{g'(x)}{x} \, dx =\ln \big(\tfrac{1+\overline{\lambda}}{1-\underline{\lambda}}\big).
\end{split}
\end{equation}
Here we have $g'(\ux)=0$ because $C(\ux,\Gamma(\ux))=0$. 

\bigskip

{\bf Claim 5: } $g$ has the following properties:\\
(i) $g(x)>0$  for $x\in(\ux,\ox)$ and $g'(x)>0$ for $x\in[\ux,\ox]$.\\
(ii) $g'(x)/x>0$ for $x\in(\ux,\ox)$.\\
(iii) $q\, g(x)\big(g'(x)+1)- (1+q)x g'(x)>0$ for $x\in[\ux,\ox]$. \\
(iv) $g(x)-x g'(x)>0$  for $x\in[\ux,\ox]$.

(Proof of Claim 5): (i) and (ii) are obvious from the construction. 

(iii) Observe that $q\, g(\ux)\big(g'(\ux)+1)- (1+q)\ux g'(\ux)=q\, g(\ux)>0$. Suppose that there exists $x_0\in[\ux,\ox]$ s.t. $q\, g(x_0)\big(g'(x_0)+1)- (1+q)x_0 g'(x_0)=0$. Then, 
$$g'(x_0)=\tfrac{q \, g(x_0)}{(1+q)x_0 - q\, g(x_0)} \quad \Longrightarrow \quad F(x_0,g(x_0))=\tfrac{q \, g(x_0)}{(1+q)x_0 - q\, g(x_0)} \quad \Longrightarrow \quad g(x_0)=\tfrac{x_0}{q} \quad \Longrightarrow \quad g'(x_0)=\tfrac{1}{q},$$ 
which contradicts to Claim 2. 

(iv) Obviously, $g(\ux)-\ux g'(\ux)=g(\ux)>0$. Suppose that there exists $x_0\in[\ux,\ox]$ such that $g(x_0)-x_0g'(x_0)=0$. We set $k=g'(x_0)$ and $g(x_0)=k x_0$. Then, $F(x_0,k x_0)=k$ can be rewritten as
\begin{displaymath}
\begin{split}
&\sqrt{B(x_0,kx_0)^2-4A(x_0,k x_0)B(x_0, kx_0)} \\
&=-\tfrac{4k(1+k)(1+q)^3 (1-kq)\sigma_2^2 \big( k^2\mu_2^2(1+q)^2 - 2k^2 \mu_2(1+q)\rho \sigma_1 \sigma_2 + ((1+k)^2-(1+2k)\rho^2)\sigma_1^2 \sigma_2^2 \big)}{\big(k^2 \mu_2^2 (1+q)^2 (kq-1)+2k\mu_2(1+q)(kq-1)\rho \sigma_1 \sigma_2- (2k(1+k)(1+q)(\delta k (1+q)-\mu_1)+(\rho^2-1+k(k+q+kq-q\rho^2))\sigma_1^2)\sigma_2^2\big)^2}
\end{split}
\end{displaymath}
Observe that $1-kq>0$ by Claim 2. Then the above equality is a contradiction because
\begin{displaymath}
\begin{split}
&\Delta_k\big( k^2\mu_2^2(1+q)^2 - 2k^2 \mu_2(1+q)\rho \sigma_1 \sigma_2 + ((1+k)^2-(1+2k)\rho^2)\sigma_1^2 \sigma_2^2 \big)\\
&=-4(1-\rho^2)\sigma_1^2\sigma_2^2((1+q)\mu_2-\rho \sigma_1 \sigma_2)^2<0\\
&\Longrightarrow \quad k^2\mu_2^2(1+q)^2 - 2k^2 \mu_2(1+q)\rho \sigma_1 \sigma_2 + ((1+k)^2-(1+2k)\rho^2)\sigma_1^2 \sigma_2^2>0 \textrm{  for any $k$}.
\end{split}
\end{displaymath}
(End of the proof of Claim 5).

\bigskip

Finally, the definition of $F$ and Claim 5 imply that $g'(x)=F(x,g(x))$ solves \eqref{g1}. Also, Claim 5 implies \proref{FBO} (4) and (5).
\end{proof}

\bigskip

Out next task is to study the case of $\delta \leq  \tfrac{q}{2(1-\rho^2)} \big( (\tfrac{\mu_1}{\sigma_1})^2+ (\tfrac{\mu_2}{\sigma_2})^2- 2\rho \tfrac{\mu_1 \mu_2}{\sigma_1 \sigma_2} \big) $. Recall from the original Merton's problem that  if $\delta >  \tfrac{q}{2(1-\rho^2)} \big( (\tfrac{\mu_1}{\sigma_1})^2+ (\tfrac{\mu_2}{\sigma_2})^2- 2\rho \tfrac{\mu_1 \mu_2}{\sigma_1 \sigma_2} \big) $, then the optimization problem is well-posed even with zero transaction costs. In case $\delta \leq  \tfrac{q}{2(1-\rho^2)} \big( (\tfrac{\mu_1}{\sigma_1})^2+ (\tfrac{\mu_2}{\sigma_2})^2- 2\rho \tfrac{\mu_1 \mu_2}{\sigma_1 \sigma_2} \big) $, it turns out that the size of the transaction costs  should be large enough to ensure the existence of the solution of the free boundary ODE. 

Before we prove \proref{hard2}, we do some preliminary analysis. For convenience, we define the functions $K$, $Q_0$, $Q_1$, $Q_2$ and $Q_3$ by
\begin{equation}
\begin{split}\label{KQT}
K(x,y,k)&:=A(x,y)k^2 + B(x,y) k + C(x,y)\\
Q_0(k)&:=4(1+k)^2 q^2 \sigma_2^2\\
Q_1(k)&:=4(1+k)q \sigma_2 \Big( 2q  (\mu_1 \sigma_2 - \rho \mu_2 \sigma_1)(1-k q) - 2(1+q)^2 \sigma_2 \big( \delta -\tfrac{q(2\mu_1(1+q)-\sigma_1^2)}{2(1+q)^2} \big) k  \Big)\\
Q_2(k)&:=\Big(2\mu_2 q (1-k q ) \sigma_1 + \rho \sigma_2 \big(2\delta k (1+q)^2 + q(k \sigma_1^2 - 2(1+k)\mu_1)  \big)    \Big)^2\\
&\qquad + (1-\rho^2)\sigma_2^2 \Big( \big(2\delta k (1+q)^2 - 2(1+k)q \mu_1 + k q \sigma_1^2  \big)^2 - 8 \delta q (1-k q)^2 \sigma_1^2     \Big)\\
Q_3(k)&:= k \big( 2\mu_2 q^2 \rho \sigma_1 \sigma_2 - 2\mu_2^2 q (1+q)^2 + (2\delta (1+q)^2 + q(2\mu_1 - \sigma_1^2))\sigma_2^2 \big) \\
&\qquad +2q \sigma_2(\mu_1 \sigma_2 - \rho \mu_2 \sigma_1) 
\end{split}
\end{equation}
To study the level curve $K(x,y,k)=0$, we define $T_u$ and $T_d$ by
\begin{equation}
\begin{split}\label{T}
T_u(x,k)&:= \tfrac{ Q_3(k) x   +2(1+k)q \sigma_2^2 + \sigma_2 \sqrt{Q_2(k) x^2 + Q_1(k) x + Q_0(k)} }{2(1+k)q(1+q)(2\delta \sigma_2^2 - q \mu_2)}\\
T_d(x,k)&:= \tfrac{ Q_3(k) x   +2(1+k)q \sigma_2^2 - \sigma_2 \sqrt{Q_2(k) x^2 + Q_1(k) x + Q_0(k)} }{2(1+k)q(1+q)(2\delta \sigma_2^2 - q \mu_2)}
\end{split}
\end{equation}
To describe limiting behaviors of $T_u$ and $T_d$, we define $l_u$ and $l_d$ by
\begin{equation}
\begin{split}\label{luld}
l_u(k)&:= \tfrac{ Q_3(k)  + \sigma_2 \sqrt{Q_2(k)} }{2(1+k)q(1+q)(2\delta \sigma_2^2 - q \mu_2)}, \quad l_d(k):= \tfrac{ Q_3(k)  - \sigma_2 \sqrt{Q_2(k)} }{2(1+k)q(1+q)(2\delta \sigma_2^2 - q \mu_2)}\\
\end{split}
\end{equation}
The following technical lemma is useful for the proof of \proref{hard2}.

\begin{lemma}\label{wellposedlemma}
Assume that $\delta \leq  \tfrac{q}{2(1-\rho^2)} \big( (\tfrac{\mu_1}{\sigma_1})^2+ (\tfrac{\mu_2}{\sigma_2})^2- 2\rho \tfrac{\mu_1 \mu_2}{\sigma_1 \sigma_2} \big) $. Let $F$ be as in  \proref{hard}, i.e., 
$$F(x,y)=\tfrac{2C(x,y)}{-B(x,y)+\sqrt{B(x,y)^2-4 A(x,y) C(x,y)}}.$$
(1) $Q_2(k)$ is quadratic in $k$ and $Q_2(k)=0$ has two distinct non-negative real roots. We define $k^*$ as the smaller root of $Q_2(k)=0$. \\
%Then $k^* \geq 0$ and $k^*=0$ iff $\delta=\tfrac{q}{2(1-\rho^2)} \big( (\tfrac{\mu_1}{\sigma_1})^2+ (\tfrac{\mu_2}{\sigma_2})^2- 2\rho \tfrac{\mu_1 \mu_2}{\sigma_1 \sigma_2} \big)$.
(2) For $x>0$ and $0\leq k \leq k^*$, $Q_2(k) x^2 + Q_1(k) x + Q_0(k) >0$.\\
(3) For $0 \leq k \leq k^*$,
$$
\big \{ (x,y): x>0, \,\,C(x,y) \geq 0, \,\, k=F(x,y) \big \} = \big \{ (x,y): x>0, \,\,C(x,y) \geq 0,\,\,  y=T_u(x,k) \textrm{  or  } T_d(x,k) \big \}  
$$
(4) The set  $\big \{ (x,y): x>0,  \,\,C(x,y) \geq 0, \,\, k=F(x,y) \big \}$ is bounded if $k > k^*$. It is unbounded if $0\leq k \leq k^*$.\\
(5) For $0 \leq k \leq k^*$, $l_d(k)>k$.\\
(6) For $0 \leq k<k^*$, we have
\begin{equation}
\begin{split}
\lim_{x\to \infty} \tfrac{\partial}{\partial x} T_{u,d} (x,k) = l_{u,d}(k) \,\, \textrm{  and  } \,\,  \lim_{x\to \infty} \tfrac{1}{x} \tfrac{\partial}{\partial k} T_{u,d} (x,k) = l_{u,d}'(k) 
\end{split}
\end{equation}
(7) There exists a constant $c > 0$ such that for all $x>c$ and $0\leq k <k^*$ the following inequalities hold
\begin{equation}
\begin{split}\label{dominator}
\Big|  \tfrac{  \tfrac{\partial}{\partial k} T_{u} (x,k)} {x(k- \tfrac{\partial}{\partial x} T_{u} (x,k))}  \Big|< c+ \tfrac{c}{\sqrt{k^*-k}}, 
\quad \Big|  \tfrac{  \tfrac{\partial}{\partial k} T_{d} (x,k)} {x(k- \tfrac{\partial}{\partial x} T_{d} (x,k))}  \Big|< c + \tfrac{c}{\sqrt{k^*-k}}.
\end{split}
\end{equation}
\end{lemma}
\begin{proof}
(1) As 
$$\Delta_k(Q_2(k)) = 32 q (1-\rho^2) \sigma_1^2 (1+q)^2 \big(\delta -\tfrac{q(2\mu_1(1+q)-\sigma_1^2)}{2(1+q)^2}\big)^2 \sigma_2^2 (2\delta \sigma_2^2 - q \mu_2^2)>0,$$
the equation $Q_2(k)=0$ has two distinct roots. To show that these roots are non-negative, it is enough to check that $Q_2(0)\geq 0, \, Q_2'(0)<0, \, Q_2(\tfrac{1}{q})>0$ and $Q_2'(\tfrac{1}{q})>0$. Indeed, 
\begin{displaymath}
\begin{split}
&Q_2(0)=8 q \sigma_1^2 \sigma_2^2 (1-\rho^2) \Big( \tfrac{q}{2(1-\rho^2)} \big( (\tfrac{\mu_1}{\sigma_1})^2+ (\tfrac{\mu_2}{\sigma_2})^2- 2\rho \tfrac{\mu_1 \mu_2}{\sigma_1 \sigma_2} \big) - \delta \Big) \geq 0,\\
&Q_2(\tfrac{1}{q})= \tfrac{ 2 (1+q)^2  \sigma_2^2 \big(  \delta-\tfrac{q(2\mu_1(1+q)-\sigma_1^2)}{2(1+q)^2}  \big)^2 }{q^2}>0,\\
&Q_2'(\tfrac{1}{q}) = \tfrac{4 \sigma_2 (1+q)^2 \big(  \delta-\tfrac{q(2\mu_1(1+q)-\sigma_1^2)}{2(1+q)^2}  \big) \big(  2\sigma_2(1+q)^2 \big(  \delta-\tfrac{q(2\mu_1(1+q)-\sigma_1^2)}{2(1+q)^2}  \big) + 2 q^2 (\mu_1 \sigma_2 - \rho \mu_2 \sigma_1)  \big)   }{q}>0.
\end{split}
\end{displaymath}
Also observe that $Q'_2(0)$ is linear in $\delta$ and 
\begin{displaymath}
\begin{split}
&Q_2'(0)= \begin{cases} 
-\tfrac{8 q^3 ((1-\rho^2)(\mu_1(1+q)-\sigma_1^2)^2 \sigma_2^2 + (1+q)^2 (\mu_2 \sigma_1 - \rho \mu_1 \sigma_2)^2)}{(1+q)^2}, \quad \textrm{when  }\delta=\tfrac{q(2\mu_1(1+q)-\sigma_1^2)}{2(1+q)^2}\\
-\tfrac{4 q^2(\mu_1 - \tfrac{\rho \mu_2 \sigma_1}{\sigma_2}) ((1-\rho^2)(\mu_1(1+q)-\sigma_1^2)^2 \sigma_2^2 + (1+q)^2 (\mu_2 \sigma_1 - \rho \mu_1 \sigma_2)^2)}{(1-\rho^2)\sigma_1^2}, \\
\qquad\qquad\qquad\qquad\qquad\qquad\qquad   \textrm{when  } \delta=  \tfrac{q}{2(1-\rho^2)}\big( (\tfrac{\mu_1}{\sigma_1})^2+ (\tfrac{\mu_2}{\sigma_2})^2- 2\rho \tfrac{\mu_1 \mu_2}{\sigma_1 \sigma_2} \big)
\end{cases}
\end{split}
\end{displaymath}
Therefore, $Q'_2(0)<0$ under the assumption of $\tfrac{q(2\mu_1(1+q)-\sigma_1^2)}{2(1+q)^2}<\delta\leq  \tfrac{q}{2(1-\rho^2)}\big( (\tfrac{\mu_1}{\sigma_1})^2+ (\tfrac{\mu_2}{\sigma_2})^2- 2\rho \tfrac{\mu_1 \mu_2}{\sigma_1 \sigma_2} \big)$.

(2) As $Q_0(k)>0$ and $Q_2(k) \geq 0 $ for $ 0 \leq k \leq k^*$, it is enough to check that $Q_1(k) \geq 0$ for $0 \leq k \leq k^*$. We may rewrite $Q_1(k)$ as
$$
Q_1(k)=4(1+k)q \sigma_2 \Big( 2q \sigma_2 (\mu_1 - \tfrac{\rho \mu_2 \sigma_1}{\sigma_2}) - \Big(  2q^2 \sigma_2 (\mu_1 - \tfrac{\rho \mu_2 \sigma_1}{\sigma_2}) + 2(1+q)^2 \sigma_2 \big( \delta -\tfrac{q(2\mu_1(1+q)-\sigma_1^2)}{2(1+q)^2} \big)    \Big ) k  \Big),
$$
and observe that $Q_1(k)\geq 0$ for $0 \leq k \leq k_1$, where
$$k_1= \tfrac{2q \sigma_2 (\mu_1 - \tfrac{\rho \mu_2 \sigma_1}{\sigma_2})}{2q^2 \sigma_2 (\mu_1 - \tfrac{\rho \mu_2 \sigma_1}{\sigma_2}) + 2(1+q)^2 \sigma_2 \big( \delta -\tfrac{q(2\mu_1(1+q)-\sigma_1^2)}{2(1+q)^2} \big) }.
$$
Observe that
$$Q_2(k_1)=- \tfrac{4q(1-\rho^2)\sigma_1^2 \big( 2\delta(1+q)^2 -q(2\mu_1(1+q)-\sigma_1^2)  \big)^2 \sigma_2^2 (2\delta \sigma_2^2 - q \mu_2^2)  }{\big(-2\rho q^2 \mu_2 \sigma_1 + \big( 2\delta(1+q)^2 -q(2\mu_1(1+q)-\sigma_1^2)  \big) \sigma_2  \big)^2}<0,
$$
and this implies that $k^*< k_1$. We conclude that $Q_1(k) \geq 0$ for $0\leq k \leq k^*$. 

(3) By the same way as in the proof of \lemref{common}, we can check that for $0 \leq k \leq k^*$, 
\begin{displaymath}\label{no1}
\big \{ (x,y): x>0, \,\, 2C(x,y) + B(x,y)k = 0, \,\, K(x,y,k)=0 \big \} = \emptyset.
\end{displaymath}
Straight-forward computations show that $y=T_u(x,k)$ and $y=T_d(x,k)$ are the solutions of $K(x,y,k)=0$. Therefore, for $0 \leq k \leq k^*$,
\begin{displaymath}
\begin{split}
&\big \{ (x,y): x>0, \,\,C(x,y) \geq 0, \,\, k=F(x,y) \big \} \\
&=\big \{ (x,y): x>0, \,\,C(x,y) \geq 0,\,\, 2C(x,y) + B(x,y)k >0, \,\, K(x,y,k)=0 \big \} \\
&= \big \{ (x,y): x>0, \,\,C(x,y) \geq 0,\,\,  y=T_u(x,k) \textrm{  or  } T_d(x,k) \big \}.
\end{split}
\end{displaymath}

(4) By the same way as in the proof of \lemref{common}, we can check that for $0 \leq k \leq k^*$, 
$$\big \{ (x,y): x>0, \,\, C(x,y)= 0, \,\, K(x,y,k)=0 \big \} = \emptyset.$$
At $x=0$, we compare the slopes of level curve $\{(x,y): C(x,y)=0\}$ with $\tfrac{\partial}{\partial x} T_u(0,k)$ and $\tfrac{\partial}{\partial x} T_d(0,k)$ to conclude that 
$$
\textrm{  for  }0 \leq k \leq k^*, \quad y=T_u(x,k) \textrm{  or  } T_d(x,k) \textrm{  implies  } C(x,y) \geq 0.
$$
Combine the above observation with part (2) and (3) of this lemma, we conclude that $\big \{ (x,y): x>0,  \,\,C(x,y) \geq 0, \,\, k=F(x,y) \big \}$ is unbounded if $0 \leq k \leq k^*$.

Now assume that $k>k^*$. In the proof of part (1), we observe that for any small enough $\epsilon>0$, $Q_2(k^*+\epsilon)<0$. For any $k \geq k^*+ \epsilon$,
\begin{displaymath}
\begin{split}
&\big \{ (x,y): x>0,  \,\,C(x,y) \geq 0, \,\, k=F(x,y) \big \} \\
& \subset \big \{ (x,y): x>0,  \,\,C(x,y) \geq 0, \,\, k^*+\epsilon \leq  F(x,y) \big \} \\
& \subset \big \{ (x,y): x>0, \,\,  T_d(x,k^*+\epsilon) \leq y \leq T_u(x,k^*+\epsilon)\big \}
\end{split}
\end{displaymath}
The last set above is part of an ellipse (because $Q_2(k^*+\epsilon)<0$), which is bounded. As $\epsilon>0$ can be chosen arbitrary small, we conclude the boundedness for $k>k^*$.

(5) As $l_d(k) - k =  \tfrac{ Q_3(k) - 2k(1+k)q(1+q)(2\delta \sigma_2^2 - q \mu_2)  - \sigma_2 \sqrt{Q_2(k)} }{2(1+k)q(1+q)(2\delta \sigma_2^2 - q \mu_2)}$,
it is enough to show that for $0 \leq k \leq k^*$,
\begin{equation}
\begin{split}\label{Q4Q5}
&Q_4(k):=Q_3(k) - 2k(1+k)q(1+q)(2\delta \sigma_2^2 - q \mu_2)>0,\\
 \textrm{  and  } \quad  &Q_5(k):=Q_4(k)^2- \sigma_2^2 Q_2(k) >0.
\end{split}
\end{equation}
To deal with the first inequality, we observe that $Q_4(k)$ is quadratic in $k$ and 
\begin{displaymath}
\begin{split}
&Q_4(0)=2q \sigma_2 (\mu_1 \sigma_2 - \rho \mu_2 \sigma_1) >0,\\
&Q_4(k_1)=\tfrac{8q(1+q)^3 \big( \delta -\tfrac{q(2\mu_1(1+q)-\sigma_1^2)}{2(1+q)^2} \big) \sigma_2 (2 \delta \sigma_2^2 - q \mu_2^2)(\mu_1 \sigma_2 - \rho \mu_2 \sigma_1)}
{\big( \sigma_2 (2\delta (1+q)^2 - q(2\mu_1 - \sigma_1^2)) - 2\mu_2 q^2 \rho \sigma_1   \big)^2}>0 
\end{split}
\end{displaymath}
where $k_1$ be as in the proof of part (2) above. Then $Q_4''(k)=-4q(1+q)(2\delta \sigma_2^2 - q \mu_2^2) <0$  implies that $Q_4(k)>0$ for $0\leq k \leq k_1$. As $k_1>k^*$, we obtain the first inequality in \eqref{Q4Q5}. 

To show the second inequality, we define $Q_6(k)$ as
$$
Q_6(k):= \tfrac{Q_5(k)}{4q(1- q k ) (2 \delta \sigma_2^2 - q \mu_2^2)}
$$
As $k^*< \tfrac{1}{q}$ (this can be seen in the proof of part (1)), it is enough to show that  $Q_6(k)>0$ for $0\leq k \leq k^*$. We can check that $Q_6(k)$ is a cubic polynomial in $k$, $Q_6(k^*)>0$, $Q_6(0)>0$ and $Q_6'''(0)<0$. Furthermore, direct computation shows that $Q_6'(0)\leq 0 $ implies that $Q_6''(0)<0$. Considering possible shapes of a graph of this cubic polynomial, we conclude that $Q_6(k)>0$ for $0\leq k \leq k^*$. Consequently, the second inequality in \eqref{Q4Q5} holds.

(6) The convergences can be shown by a straightforward calculation.

(7) In part (6), the direct computation also yields that the convergence $\lim_{x\to \infty} \tfrac{\partial}{\partial x} T_{d} (x,k) = l_{d}(k)$ is uniform in $k \in [0,k^*]$. This observation, together with part (5), implies that there exists $\epsilon>0$ and $c>0$ such that $\tfrac{\partial}{\partial x} T_{d} (x,k)-k >\epsilon$ for all $x>c$ and $k\in [0,k^*]$. Also, we observe that for $k\in [0,k^*)$ and $x>1$, 
$$
  \tfrac{1}{x} \tfrac{\partial}{\partial k} T_{d} (x,k) < const \cdot \big(1+ \tfrac{1}{\sqrt{Q_2(k)}} \big) \leq const \cdot\big(1+ + \tfrac{1}{\sqrt{k^*-k}}\big),
$$   
where $const$s are generic constants independent of $x$ and $k$. Now we obtain the first inequality in \eqref{dominator} by choosing $c$ large enough. The second inequality can be shown in the same way.
\end{proof}

Now we define the constant $c^* = c^*(\mu_1,\mu_2,\sigma_1,\sigma_2,\delta, p, \rho)$ to describe the well-poseness.

\begin{definition}\label{c definition}
The constant $c^*$ is defined as follows
$$
c^*=\int_0^{k^*} k \Big( \frac{l_u'(k)}{k- l_u(k) } - \frac{l_d'(k)}{k- l_d(k) }    \Big)   dk,
$$
where $k^*$ is as in \lemref{wellposedlemma} (1). 
\end{definition}

\begin{remark}
By \lemref{wellposedlemma} (6) and (7), we observe that  
$$\Big|k \big( \tfrac{l_u'(k)}{k- l_u(k) } - \tfrac{l_d'(k)}{k- l_d(k) }    \big)  \Big| \leq c+ \tfrac{c}{\sqrt{k^*-k}}.$$
Hence $c^*$ is well-defined. And $c^*\geq 0$, because for $k\in [0,k^*)$, we have $l_u'(k)<0$, $l_d'(k)>0$ and $l_u(k)>l_d(k)>k$. Moreover, 
$$c^*=0 \iff k^*=0 \iff \delta=  \tfrac{q}{2(1-\rho^2)}\big( (\tfrac{\mu_1}{\sigma_1})^2+ (\tfrac{\mu_2}{\sigma_2})^2- 2\rho \tfrac{\mu_1 \mu_2}{\sigma_1 \sigma_2} \big).$$
\end{remark}

We are ready to prove \proref{hard2}. The main idea for the proof is similar to that of Proposition 6.13 in \cite{ChoSirZit11}.

\begin{proposition}\label{hard2} If $0<p<1$, $\mu_1>\frac{\rho \mu_2 \sigma_1}{\sigma_2}$, $\delta \leq  \tfrac{q}{2(1-\rho^2)} \big( (\tfrac{\mu_1}{\sigma_1})^2+ (\tfrac{\mu_2}{\sigma_2})^2- 2\rho \tfrac{\mu_1 \mu_2}{\sigma_1 \sigma_2} \big) $ and 
$c^*< \ln \big(\tfrac{1+\overline{\lambda}}{1-\underline{\lambda}}\big)$, then \proref{FBO} holds.
\end{proposition}

\begin{proof}

As for Claim 1 in the proof of \proref{hard}, we can show that for $0<a<\infty$, there exists a function $g_a:[a, b_a] \mapsto \bR$ such that 
\begin{displaymath}
\begin{split}
g_a'(x)=F(x,g_a(x)) \textrm{  and  }  g_a(a)=\Gamma(a)  \\
\end{split}
\end{displaymath}
with $b_a:=\inf \{x>a : \, C(x,g_a(x))=0 \}$, where $F$ and $\Gamma$ are as in \eqref{FG}. Because the level curve $C=0$ is hyperbola (due to the condition $\delta \leq  \tfrac{q}{2(1-\rho^2)} \big( (\tfrac{\mu_1}{\sigma_1})^2+ (\tfrac{\mu_2}{\sigma_2})^2- 2\rho \tfrac{\mu_1 \mu_2}{\sigma_1 \sigma_2} \big) $), $b_a$ might be $\infty$. The next claim is useful for the proof of $b_a<\infty$. 

\bigskip

{\bf Claim 1}: \emph{$g_a'(x)$ does not admit a local minimum on $(a, b_a)$.}

(Proof of Claim 1): Suppose that $g_a'$ has a local minimum point $x_m\in (a, b_a)$. Then, there exists $\epsilon>0$ such that $g_a'(x_m) \leq g_a'(x) \textrm{  for  } x\in [x_m- \epsilon, x_m+\epsilon]$. Let $k_m:= g_a'(x_m)$. By \lemref{wellposedlemma} and the construction of $T_u$ and $T_d$, we observe that
\begin{equation}
\begin{split}\label{no local}
&g_a(x_m)= T_u(x_m,k_m) \textrm{  or  } T_d(x_m, k_m),\\
&T_d(x,k_m) \leq g_a(x) \leq T_u(x,k_m) \textrm{  for  } x\in [x_m-\epsilon , x_m+ \epsilon].
\end{split}
\end{equation}
Suppose that $g_a(x_m)=T_u(x_m,k_m)$. By \eqref{no local}, $g_a$ and $T_u(\cdot, k_m)$ should tangent at $x=x_m$, i.e., $\tfrac{\partial}{\partial x} T_u(x_m,k_m) = g_a'(x_m) = k_m$. The observation $\tfrac{\partial^2}{\partial x^2} T_u(x,k_m)<0$, together with \eqref{no local}, implies that $g_a'(x) >  \tfrac{\partial}{\partial x} T_u(x,k_m)$ for $x\in (x_m, x_m+\epsilon]$. This lead to the following contradiction
$$
0<\int_{x_m}^{x_m+\epsilon} \Big( g_a'(x) - \tfrac{\partial}{\partial x} T_u(x,k_m) \Big) dx = g_a(x_m+\epsilon) - T_u(x_m+\epsilon, k_m) \leq 0,
$$
where the last inequality is due to \eqref{no local}. Likewise, the case of $g_a(x_m)=T_d(x_m,k_m)$ also ends up with a contradiction.\\
(End of the proof of Claim 1).

\bigskip

{\bf Claim 2}: \emph{$b_a<\infty$ and $g_a'(b_a)=0$.}

(Proof of Claim 2): Suppose that $b_a=\infty$. Claim 1 implies that there exists $k_0$ such that $\lim_{x\to \infty} g_a'(x)=k_0$. \lemref{wellposedlemma} (4) and (5) imply that $0 \leq k_0 \leq k^*$ and $l_d(k_0 - \epsilon) > k_0 + 2\epsilon$ for some $\epsilon>0$. By \lemref{wellposedlemma} (6), we observe that $\tfrac{\partial}{\partial x} T_d(x,k_0-\epsilon)> g_a'(x)+\epsilon$ for large enough $x$. But this leads to a contradiction, because $g_a'(x)>k_0 - \epsilon$ for large enough $x$ implies that $g_a(x)>T_d(x, k_0 - \epsilon)$ for large enough $x$. 

For $g_a'(b_a)=0$, we simply observe that $F(x,y)=0$ for $(x,y)$ such that $C(x,y)=0$ and $x>0$. \\
(End of the proof of Claim 2).

\bigskip

{\bf Claim 3}: \emph{Let $G(a):=\int_a^{b_a} \frac{g_a'(x)}{x} dx $. Then $G$ has the following properties: \\
(i) $G$ is continuous on $(0,\infty)$.\\
(ii) $\lim_{a\uparrow \infty} G(a)=c^*$.\\
(iii) $\lim_{a\downarrow 0} G(a)=\infty$.
}

(Proof of Claim 3): The proofs of parts (i) and (iii) are the same as that for Claim 4 in \proref{hard}. To prove part (ii), we define $x_u(a)$ and $x_d(a)$ as
$$
x_u(a) := \{x>0: g_a(x) = T_u(x,k^*)\} \quad \textrm{and} \quad x_d(a) := \{x>0: g_a(x) = T_d(x,k^*)\}.
$$
\begin{figure}
\begin{center}
  \includegraphics[width=6cm]{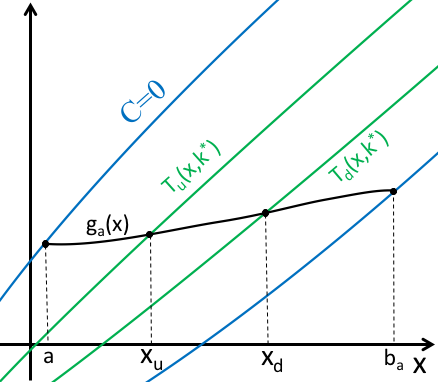}
\end{center}
\figcaption{$x_u$ and $x_d$}
\end{figure}
See Figure 2 for the illustration. \lemref{wellposedlemma} (4) implies that $x_u$ and $x_d$ are well-defined and $a<x_u(a)<x_d(a)<b_a$. Claim 1 implies that $g_a'(x)$ strictly increases for $x\in (a,x_u(a)]$ and strictly decreases for $x\in [x_d(a),b_a)$. Therefore, there exists the inverse function $I_a:[0,k^*] \to [a,x_u(a)]$ of $g_a$, i.e., $g_a'(I_a(k))=k$ for $k\in [0,k^*]$. Then, we observe that
$$
g_a(I_a(k))=T_u(I_a(k),k) \textrm{  and  } I_a'(k)=\tfrac{  \frac{\partial}{\partial k} T_{u} (x,k)} {k- \frac{\partial}{\partial x} T_{u} (x,k)},
$$
where the second equality can be obtained by differentiating the first equality. By changing variable as $x=I_a(k)$,
\begin{displaymath}
\begin{split}
\int_{a}^{x_u(a)} \tfrac{g_a'(x)}{x} dx = \int_0^{k^*} \tfrac{k}{I_a(k)} \cdot \tfrac{  \frac{\partial}{\partial k} T_{u} (I_a(k),k)} {k- \frac{\partial}{\partial x} T_{u} (I_a(k),k)} dk \longrightarrow \int_0^{k^*} \tfrac{ k l_u'(k)}{k-l_u(k)} dk \quad \textrm{as  }a\to \infty,
\end{split}
\end{displaymath}
where the convergence is justified by \lemref{wellposedlemma} (6) and (7) and the observation $\lim_{a\to \infty} I_a(k) = \infty$. In the same way, we can check that
\begin{displaymath}
\begin{split}
\int_{x_d(a)}^{b_a} \tfrac{g_a'(x)}{x} dx \longrightarrow \int_0^{k^*} -\tfrac{ k l_d'(k)}{k-l_d(k)} dk \quad \textrm{as  }a\to \infty.
\end{split}
\end{displaymath}
Therefore, to complete the proof of (ii), it remains to prove that $\lim_{a\to \infty} \int_{x_u(a)}^{x_d(a)} \tfrac{g_a'(x)}{x} dx = 0 $. \proref{wellposedlemma} (5) and (6) implies that there exists $\epsilon>0$ and $x_\epsilon>0$ such that $\tfrac{\partial}{\partial x} T_d(x,k^*) > k^* + 2 \epsilon$ for $x>x_\epsilon$. By \lemref{wellposedlemma} (4), we can find $a_\epsilon>0$ such that $g_a'(x)<k^*+\epsilon$ for $a>a_\epsilon$ and $x\in [a,b_a]$. If $a>\max\{a_\epsilon, x_\epsilon\}$, then
\begin{equation}
\begin{split}\label{lastguy}
\epsilon \, (x_d(a)-x_u(a)) & \leq \int_{x_u(a)}^{x_d(a)} \Big(  \tfrac{\partial}{\partial x} T_d(x,k^*)  - g_a'(x)   \Big) dx = T_u(x_u(a),k^*) - T_d(x_u(a),k^*) \\
&=\tfrac{\sigma_2 \sqrt{Q_1(k^*)x_u(a) + Q_0(k^*)}} {(1+k^*)q(1+q)(2\delta \sigma_2^2  - q \mu_2)},
\end{split}
\end{equation}
where the first equality is due to $g_a(x_u(a))=T_u(x_u(a),k^*)$ and $g_a(x_d(a))=T_d(x_d(a),k^*)$, and the second equality is from the definition of $T_u,T_d$ and $k^*$.
Therefore,
\begin{displaymath}
\begin{split}
\limsup_{a\to \infty} \Big| \int_{x_u(a)}^{x_d(a)} \tfrac{g_a'(x)}{x} dx  \Big| & \leq \limsup_{a\to \infty} \Big| (k^*+\epsilon) \ln\Big( 1+ \tfrac{x_d(a)-x_u(a)}{x_u(a)}  \Big)  \Big|  \\
& \leq \limsup_{a\to \infty} \Big| (k^*+\epsilon) \ln\Big( 1+ \tfrac{\sigma_2 \sqrt{Q_1(k^*)x_u(a) + Q_0(k^*)}} {\epsilon(1+k^*)q(1+q)(2\delta \sigma_2^2  - q \mu_2)x_u(a)} \Big)  \Big|  =0,
\end{split}
\end{displaymath}
where the second inequality is from \eqref{lastguy}, and the equality holds because $\lim_{a\to \infty} x_u(a)=\infty$.\\
(End of the proof of Claim 3).

\bigskip

By Claim 3 and the intermediate value theorem, we can choose $a\in (0,\infty)$ such that 
$$\int_a^{b_a} \tfrac{g_a'(x)}{x} \, dx =\ln (\tfrac{1+\overline{\lambda}}{1-\underline{\lambda}}).$$ 

The proof of (4) and (5) of \proref{FBO} is exactly same as the proof in \proref{hard}.
\end{proof}

\begin{figure}
\begin{center}
  \includegraphics[width=6cm]{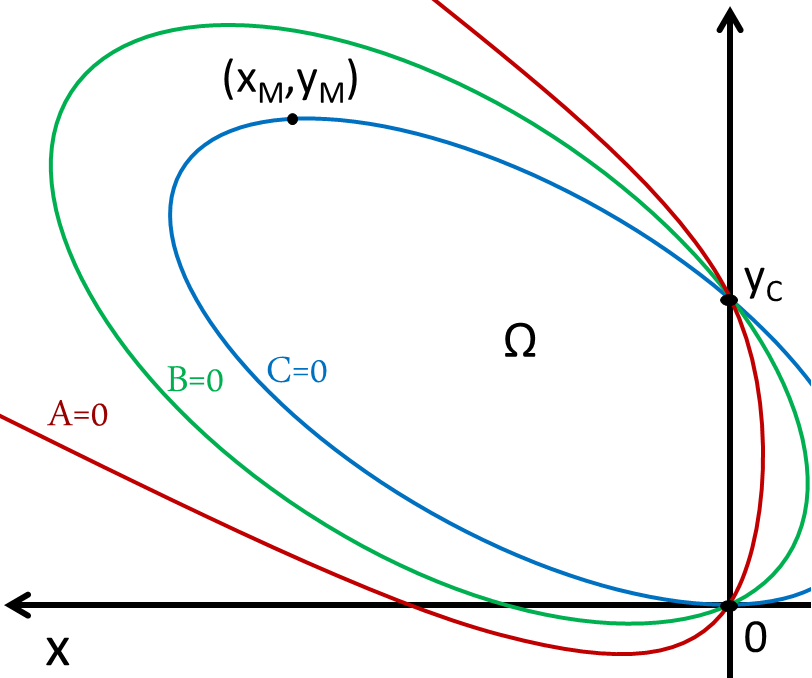}\qquad \includegraphics[width=6cm]{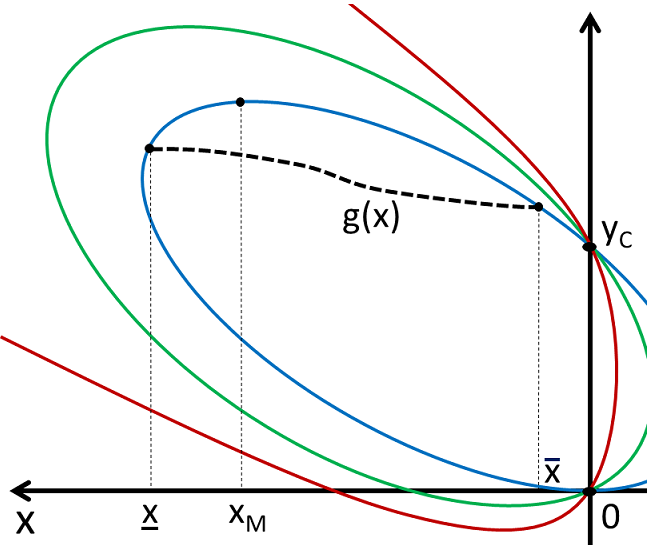}
\end{center}
\figcaption{$0<p<1$ and $\mu_1<\frac{\rho \mu_2 \sigma_1}{\sigma_2}$}
\end{figure}

\begin{proposition}\label{pro2} In case $0<p<1$ and $\mu_1<\frac{\rho \mu_2 \sigma_1}{\sigma_2}$, \proref{FBO} holds.
\end{proposition}
\begin{proof}
Considering the quadratic curves $A=0,\, B=0,\,C=0$, we define the region $\Omega$ (see Figure 3) by 
$$\Omega:=\{(x,y):\, x<0,\, C(x,y)\geq 0 ,\, B(x,y)>0,\, A(x,y)>0 \}.$$
As in \proref{hard} and \proref{hard2}, we can prove that there exist $\ux<\ox<0$ and $g\in C^2([\ux,\ox])$ such that
\begin{equation}
\begin{split}\label{g2}
&g'(x)=F(x,g(x)), \quad g'(\ux)=0,\quad g'(\ox)=0,\quad \int_{\ux}^{\ox} \frac{g'(x)}{x} \, dx =\ln \big(\tfrac{1+\overline{\lambda}}{1-\underline{\lambda}}\big),\\
&\textrm{where } F(x,y):=\tfrac{2C(x,y)}{-B(x,y)-\sqrt{B(x,y)^2-4 A(x,y) C(x,y)}}.
\end{split}
\end{equation}
Note that $F$ is different from $F$ in \proref{hard}, but the analysis is almost same. Also we can prove the following properties of $g$ by the same way as in \proref{hard}:\\
(i) $g(x)>0$ and $g'(x)<0$ for $x\in [\ux,\ox]$.\\
(ii) $g'(x)/x>0$ for $x\in (\ux,\ox)$.\\
(iii) $q\, g(x)\big(g'(x)+1)- (1+q)x g'(x)>0$ for $x\in[\ux,\ox]$.\\
(iv) $g(x)-x g'(x)>0$  for $x\in[\ux,\ox]$.

The proof is done by \eqref{g2} and (i)-(iv).  
\end{proof}

\begin{figure}
\begin{center}
  \includegraphics[width=6cm]{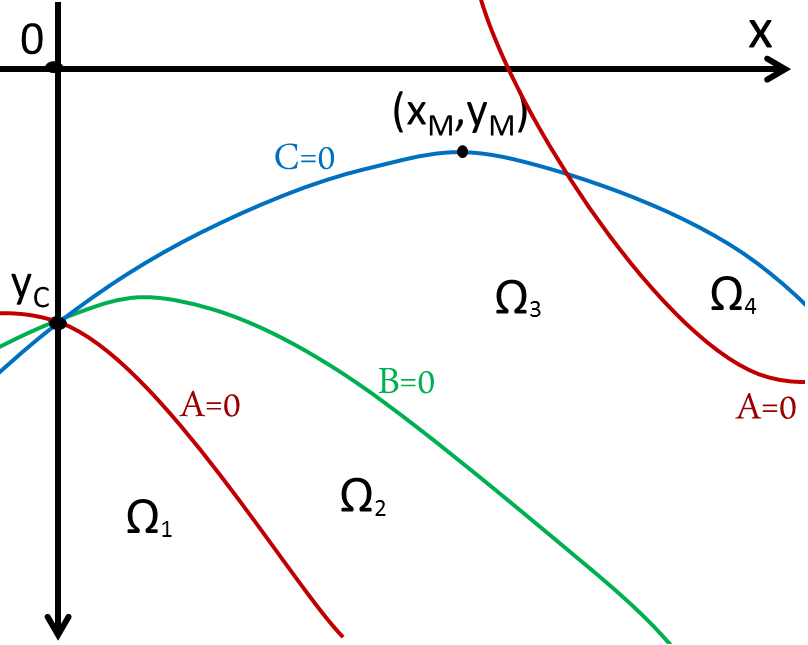}\qquad \includegraphics[width=6cm]{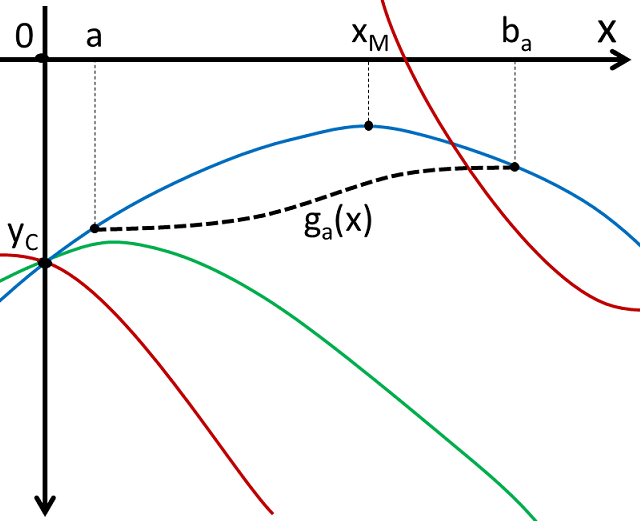}
\end{center}
\figcaption{$p<0$ and $\mu_1>\frac{\rho \mu_2 \sigma_1}{\sigma_2}$}
\end{figure}

\begin{proposition}\label{p nega} In case $p<0$ and $\mu_1>\frac{\rho \mu_2 \sigma_1}{\sigma_2}$, \proref{FBO} holds.
\end{proposition}
\begin{proof}
Because $p<0$, we have
$$q<0, \,\,1+q>0,\,\,  x_M>0,\,\, y_M<0.$$
As in \proref{hard}, we can show that the level curve $C=0$ is a hyperbola, and the quadratic curves $A=0,\, B=0,\,C=0$ are as in Figure 4 (we choose the lower curves of the hyperbola). Also,
\begin{displaymath}
\begin{split}
&\{(x,y):\, x>0,\,y<0,\, C(x,y)\geq0 \}=\Omega_1 \cup \Omega_2 \cup \Omega_3\cup \Omega_4,\\
&\textrm{where  }
\begin{cases}
\Omega_1:=\{(x,y):\, x>0,  \, y<0,\,\,  C>0,\, B\geq 0,\, A\geq 0\} \\
\Omega_2:=\{(x,y):\, x>0,  \, y<0, \,\,  C>0,\, B\geq 0,\, A< 0\} \\
\Omega_3:=\{(x,y):\, x>0,  \, y<0, \,\,  C\geq 0,\, B< 0,\, A< 0\} \\
\Omega_4:=\{(x,y):\, x>0,  \, y<0, \,\,  C\geq 0,\, B< 0,\, A\geq 0\} \\
\end{cases}
\end{split}
\end{displaymath}

%By \proref{common} (2) and (3), if $2\delta(1+q)^2 + q(\sigma_1^2-2(1+q)\mu_1)=0$, 
%$$\{(x,y):\, x \neq 0, \,\, B(x,y)^2-4A(x,y)C(x,y) \leq 0\} = \emptyset.$$
%If $2\delta(1+q)^2 + q(\sigma_1^2-2(1+q)\mu_1)\neq 0$, then $A(x_D,y_D)>0, \,B(x_D, y_D)>0$ and  $C(x_D,y_D)>0$, so $(x_D,y_D)\in \Omega_1$. 

{\bf Claim}: \emph{For any $a\in \bR$ such that $0<a<x_M$, there exist a constant $b_a>a$ and a function $g_a:[a,b_a]\mapsto \bR$ such that
\begin{displaymath}
\begin{split}
&g_a'(x)=F(x,g_a(x)), \quad g_a(a)=\Gamma(a), \quad g_a'(b_a)=0,  \\
\textrm{where } &\begin{cases} F(x,y):=\tfrac{2C(x,y)}{-B(x,y)+\sqrt{B(x,y)^2-4 A(x,y) C(x,y)}},\\
\Gamma(x):=\tfrac{-q\sigma_2 (\sigma_2+(\mu_1 \sigma_2 - \rho \sigma_1 \mu_2)x) + \sigma_2\sqrt{ 2q\delta(\rho^2-1)\sigma_1^2\sigma_2^2 x^2 +  q^2( \mu_2^2 \sigma_1^2 x^2 - 2\rho \sigma_1 \sigma_2 \mu_2 x(-1+\mu_1 x) + \sigma_2^2(-1+\mu_1 x)^2)  }     }{q(1+q)(2 \delta \sigma_2^2 - q \mu_2^2)}.
\end{cases}
\end{split}
\end{displaymath}
In fact, $y=\Gamma(x)$ is the equation of the lower curve of the hyperbola $C=0$.}

We observe that $\Gamma'(x)>0$ for $0<x<x_M$, $\Gamma'(x)<0$ for $x>x_M$ and $\lim_{x\to \infty }\Gamma(x)=-\infty$. Using this observation, we can prove Claim by the same way as in \proref{hard}.

Again, as in \proref{hard}, there exist $0<\ux<\ox$ and $g\in C^2([\ux,\ox])$ such that
\begin{equation}
\begin{split}\label{g3}
g'(x)=F(x,g(x)), \quad g'(\ux)=0,\quad g'(\ox)=0, \quad \int_{\ux}^{\ox} \frac{g'(x)}{x} \, dx =\ln \big(\tfrac{1+\overline{\lambda}}{1-\underline{\lambda}}\big),
\end{split}
\end{equation}
and $g$ satisfies the following properties:\\
(i) $g(x)<0$ and $g'(x)>0$ for $x\in [\ux,\ox]$.\\
(ii) $g'(x)/x>0$ for $x\in (\ux,\ox)$.\\
(iii) $q\, g(x)\big(g'(x)+1)- (1+q)x g'(x)>0$ for $x\in[\ux,\ox]$.\\
(iv) $g(x)-x g'(x)<0$  for $x\in[\ux,\ox]$.

The proof is done by \eqref{g3} and (i)-(iv).  
\end{proof}

\begin{figure}
\begin{center}
  \includegraphics[width=6cm]{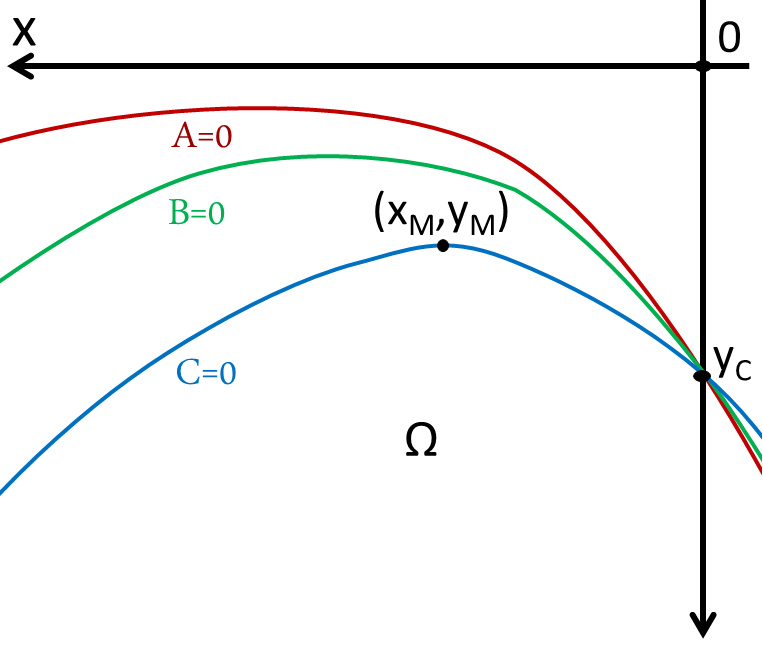}\qquad \includegraphics[width=6cm]{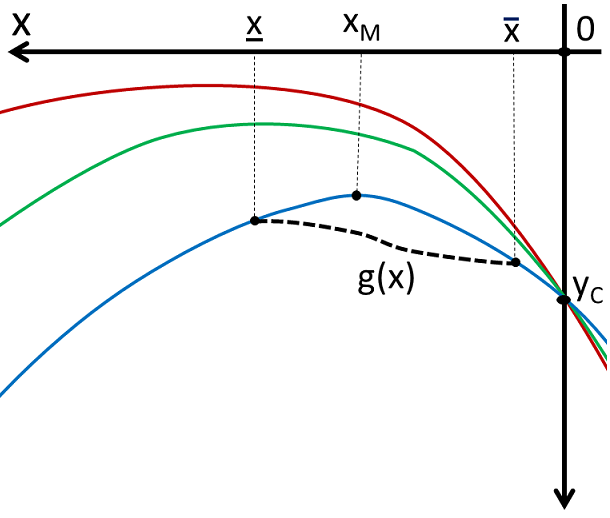}
\end{center}
\figcaption{$p<0$ and $\mu_1<\frac{\rho \mu_2 \sigma_1}{\sigma_2}$}
\end{figure}

\begin{proposition}\label{pro4} In case $p<0$ and $\mu_1<\frac{\rho \mu_2 \sigma_1}{\sigma_2}$, \proref{FBO} holds.
\end{proposition}
\begin{proof}
As in \proref{p nega}, we can show that the quadratic curves $A=0,\, B=0,\,C=0$ are as in Figure 5. The region $\Omega$ is defined as
$$\Omega:=\{(x,y):\, x<0,\, y<0, \,  C(x,y)\geq 0 ,\, B(x,y)>0,\, A(x,y)>0 \}.$$
As in \proref{p nega}, we can prove that there exist $\ux<\ox<0$ and $g\in C^2([\ux,\ox])$ such that
\begin{equation}
\begin{split}\label{g4}
&g'(x)=F(x,g(x)), \quad g'(\ux)=0,\quad g'(\ox)=0,\quad \int_{\ux}^{\ox} \frac{g'(x)}{x} \, dx =\ln \big(\tfrac{1+\overline{\lambda}}{1-\underline{\lambda}}\big),\\
&\textrm{where } F(x,y):=\tfrac{2C(x,y)}{-B(x,y)-\sqrt{B(x,y)^2-4 A(x,y) C(x,y)}},
\end{split}
\end{equation}
and $g$ satisfies the following properties:\\
(i) $g(x)<0$ and $g'(x)<0$ for $x\in [\ux,\ox]$.\\
(ii) $g'(x)/x>0$ for $x\in (\ux,\ox)$.\\
(iii) $q\, g(x)\big(g'(x)+1)- (1+q)x g'(x)>0$ for $x\in[\ux,\ox]$.\\
(iv) $g(x)-x g'(x)<0$  for $x\in[\ux,\ox]$.

The proof is done by \eqref{g4} and (i)-(iv).  
\end{proof}

\bigskip

{\bf Proof of \thmref{well-posed thm}}\\
By \thmref{main verify}, we have  (3) $\Rightarrow$ (2). \defref{shadow_define} implies that (2) $\Rightarrow$ (1). It remains to show that (1) $\Rightarrow$ (3). 

{\bf Claim}: \emph{Assume that $0<p<1$, $\mu_1>\frac{\rho \mu_2 \sigma_1}{\sigma_2}$ and $\delta \leq  \tfrac{q}{2(1-\rho^2)} \big( (\tfrac{\mu_1}{\sigma_1})^2+ (\tfrac{\mu_2}{\sigma_2})^2- 2\rho \tfrac{\mu_1 \mu_2}{\sigma_1 \sigma_2} \big) $. As in Claim 3 in \proref{hard2}, let $G(a)=\int_a^{b_a} \tfrac{g_a'(x)}{x}$. Then $G(a)$ is a strictly decreasing function.}

(Proof of Claim): Define $\overline{\lambda}(a):=(1-\underline{\lambda}) e^{G(a)} -1$, then $G(a)=\ln \big( \tfrac{1+\overline{\lambda}(a)}{ 1-\underline{\lambda}}\big)$. We consider a model with transaction costs $\overline{\lambda}(a)$, i.e., parameterize the model by $a$. In \eqref{fr}, if $(\eta_0,\eta_1,\eta_2)=(1,0,0)$, then 
$r(x)=-\tfrac{x}{q g_a(x)}<0$ for $x\in [a,b_a]$. Therefore, according to \eqref{hx}, $\hat{x}=a$. Applying \thmref{main verify} to this situation, we observe that the value of the optimization problem is $\tfrac{ \xi(\hx)^p}{p} |g_a(\hx)|^{1-p}=\tfrac{1}{p}\, g_a(a)^{1-p}=\tfrac{1}{p}\, \Gamma(a)^{1-p}$, where $\Gamma$ is as in \eqref{FG}. As $\Gamma$ is strictly increasing, the value is strictly increasing on $a$. But the value of the optimization problem should be strictly decreasing over $\overline{\lambda}$ (higher transaction cost term reduce the value). Therefore, we conclude that $\overline{\lambda}(a)$ should be a strictly decreasing function of $a$, so is $G(a)$.\\
(End of the proof of Claim).

\smallskip
In the proof of Claim above, we further observe that the value of the optimization problem with transaction costs $\overline{\lambda}(\infty):=\lim_{a\to \infty}\overline{\lambda}(a)$ should be infinity, because $\lim_{a\to \infty}\Gamma(a)=\infty$. Because $\lim_{a\to \infty}G(a)=c^* = \ln \big( \tfrac{1+\overline{\lambda}(\infty)}{ 1-\underline{\lambda}}\big)$, we conclude that if $0<p<1$, $\mu_1>\frac{\rho \mu_2 \sigma_1}{\sigma_2}$, $\delta \leq  \tfrac{q}{2(1-\rho^2)} \big( (\tfrac{\mu_1}{\sigma_1})^2+ (\tfrac{\mu_2}{\sigma_2})^2- 2\rho \tfrac{\mu_1 \mu_2}{\sigma_1 \sigma_2} \big) $  and $c^*\geq \ln \big(\tfrac{1+\overline{\lambda}}{1-\underline{\lambda}}\big)$, then the value is infinity.

In the same way, we reach the same conclusion for the case of $\mu_1<\frac{\rho \mu_2 \sigma_1}{\sigma_2}$. Therefore, the problem is ill-posed if $0<p<1$, $\delta \leq  \tfrac{q}{2(1-\rho^2)} \big( (\tfrac{\mu_1}{\sigma_1})^2+ (\tfrac{\mu_2}{\sigma_2})^2- 2\rho \tfrac{\mu_1 \mu_2}{\sigma_1 \sigma_2} \big) $  and $c^*\geq \ln \big(\tfrac{1+\overline{\lambda}}{1-\underline{\lambda}}\big)$. This is the contrapositive of (1) $\Rightarrow$ (3). $\square$

\begin{remark}
The claim in the proof of \thmref{well-posed thm} can be extended to other parameter conditions. This means that in \proref{FBO}, the solution $g$ and the free boundaries $\overline{x}$ and $\underline{x}$ are unique. 
\end{remark}

\section*{Appendix: The case of  $\mu_2= \frac{\rho \sigma_1 \sigma_2}{1+q}$ }

%\noindent{\bf $\bullet$ The case of  $\mu_2= \frac{\rho \sigma_1 \sigma_2}{1+q}$}

%\smallskip

As we see in Section 6, our optimal consumption/investment problem reduces to the analysis of the ODE in \eqref{ode}, i.e., 
$$A(x,g(x)) g'(x)^2 +B(x,g(x))g'(x) + C(x,g(x))=0.$$ 
We substitute $\mu_2= \frac{\rho \sigma_1 \sigma_2}{1+q}$ to the above ODE  and obtain 
\begin{equation}
\begin{split}\label{one-reduce}
g'(x)=-1 \quad \textrm{or} \quad g'(x)=\frac{P(x,g(x))}{Q(x,g(x))},
\end{split}
\end{equation}
where 
\begin{displaymath}
\begin{split}
P(x,y)&:=-2 q \delta^* y^2+ 2p(\mu^* x + \sgn(p))y - (1-p)^2 (\sigma^*)^2 x^2\\
Q(x,y)&:=-P(x,y) + (p(\sigma^*)^2 - 2 (1+q)\delta^*) x y + 2(\mu^*-(1-p)(\sigma^*)^2)x^2 + 2\sgn(p) x
\end{split}
\end{displaymath}
with 
$$\mu^* := \mu_1 - (1-p)\rho^2 \sigma_1^2, \quad \sigma^*:=\sigma_1\sqrt{1-\rho^2}, \quad \delta^*:=\delta - \tfrac{ p(1-p) \rho^2 \sigma_1^2}{2}.$$
\cite{ChoSirZit11} studies the optimization problem with only one illiquid asset, and ends up with the ODE $g'(x)=\tfrac{P(x,g(x))}{Q(x,g(x))}$ (cf. see Section 5 in \cite{ChoSirZit11}) which is same as the ODE in \eqref{one-reduce}. Therefore, we conclude that if $\mu_2= \frac{\rho \sigma_1 \sigma_2}{1+q}$, then the ODE in our model with two risky assets simplifies to the ODE in \cite{ChoSirZit11}. 

Alternatively, we can also find a financial explanation for why our model reduces to a simpler model, in case  $\mu_2= \frac{\rho \sigma_1 \sigma_2}{1+q}$. Using \eqref{pi12} with  $\mu_2= \frac{\rho \sigma_1 \sigma_2}{1+q}$,  we observe that the proportion of `liquid wealth' invested in the liquid risky asset is a constant
$$\tfrac{ \hat{\varphi}^{(2)}_t S_t^{(2)} }{ \hat{\varphi}^{(0)}_t+\hat{\varphi}^{(2)}_t S_t^{(2)}} = \tfrac{\pi_2(X_t)}{1-\pi_1(X_t)}=\tfrac{(1+q)\mu_2}{\sigma_2^2}.$$
In other words, economically speaking, 
\begin{displaymath}
\begin{split}
&\Big(\textrm{our model with a bond, the illiquid asset $S^{(1)}$ and the liquid asset $S^{(2)}$}\Big)\\
\iff&  \,\,\Big( \textrm{the model with the illiquid asset $S^{(1)}$ and a new liquid asset $S^{(3)}$ (without bond)} \Big) ,\\
\end{split}
\end{displaymath}
where  $dS_t^{(3)}=S_t^{(3)}\tfrac{(1+q)\mu_2}{\sigma_2^2}(\mu_2 dt + \sigma_2 dB^{(2)}_t)$. 
%This is a financial explanation why our model reduces to the model in \cite{ChoSirZit11}, in case  $\mu_2= \frac{\rho \sigma_1 \sigma_2}{1+q}$.

\bigskip

{\bf Acknowledgements.} The author has been supported by the Start-Up fund at UNIST (Grant No. 1.160099.01) and the National Research Foundation of Korea (Grant No. 2017R1E1A1A03070732).

\bibliographystyle{siam}  % Here the bibliography 		     %
\bibliography{reference}        % is inserted.			     %

\end{document}